\documentclass[conference,letterpaper]{IEEEtran}
\IEEEoverridecommandlockouts

\usepackage[utf8]{inputenc} 
\usepackage[T1]{fontenc}
\usepackage{microtype}

\usepackage{url}              % provides \url{...}
\usepackage[noadjust,space]{cite}             % improves presentation of citations

\usepackage[cmex10]{amsmath}  % Use the [cmex10] option to ensure complicance
\usepackage{mleftright}       % fix to wrong spacing of \left-,
\mleftright                   % \middle- \right-commands 

\usepackage{graphicx}         % provides \includegraphics{...} to
\usepackage{booktabs}         % fixes poor spacing in tables and
\usepackage{xparse}
\usepackage{ifthen}

\usepackage{amsfonts}
\usepackage{amssymb}
\usepackage[x11names]{xcolor}
\usepackage{mathtools}
\usepackage{bm}
\usepackage{bbm}
\usepackage{xspace}
\usepackage{mathdots}
\usepackage{float}
\usepackage{stmaryrd}
\usepackage{mdframed}
\usepackage{nicefrac}
\usepackage{soul}
\usepackage[inline]{enumitem}
\usepackage{relsize}

\usepackage{array}
\usepackage{booktabs}
\usepackage{tabularx}
\usepackage{multirow}
\usepackage{multicol}
\usepackage{makecell}
\usepackage{hhline}

\usepackage{boxedminipage}
\usepackage{adjustbox}
\usepackage[ruled,vlined,linesnumbered]{algorithm2e}

\usepackage{tikz}
\usetikzlibrary{positioning}
\usetikzlibrary{fit}
\usetikzlibrary{calc}
\usetikzlibrary{backgrounds}
\usetikzlibrary{shapes}
\usetikzlibrary{patterns}
\usetikzlibrary{matrix}
\usetikzlibrary{intersections}
\usetikzlibrary{decorations}
\usetikzlibrary{arrows.meta}

\usepackage[a-3b]{pdfx}
\hypersetup{pdfstartview=FitH,colorlinks,linkcolor=blue,filecolor=blue,citecolor=blue,urlcolor=blue,breaklinks}

\usepackage{amsthm}
\usepackage{thm-restate}
\usepackage{thmtools}
\usepackage[nameinlink,capitalize]{cleveref}

\newtheorem{corollary}{Corollary}

\newtheorem{lemma}{Lemma}

\newtheorem{definition}{Definition}

\theoremstyle{remark}

\let\leq\leqslant
\let\geq\geqslant

\newcommand{\bbN}{\ensuremath{\mathbb{N}}\xspace}

\newcommand{\bbR}{\ensuremath{\mathbb{R}}\xspace}

\newcommand{\bbZ}{\ensuremath{\mathbb{Z}}\xspace}

\newcommand{\tC}{\ensuremath{\tilde{C}}\xspace}

\newcommand{\tJ}{\ensuremath{\widetilde{J}}\xspace}

\newcommand{\cA}{\ensuremath{\mathcal{A}}\xspace}

\newcommand{\cC}{\ensuremath{\mathcal{C}}\xspace}

\newcommand{\cE}{\ensuremath{\mathcal{E}}\xspace}

\newcommand{\cG}{\ensuremath{\mathcal{G}}\xspace}

\newcommand{\cI}{\ensuremath{\mathcal{I}}\xspace}
\newcommand{\cJ}{\ensuremath{\mathcal{J}}\xspace}

\DeclareMathOperator{\polylog}{polylog}
\DeclareMathOperator{\poly}{poly}

\newcommand{\zo}{\ensuremath{{\{0,1\}}}\xspace}
\newcommand{\bin}[1]{\ensuremath{\zo^{#1}}}
\newcommand{\negl}{\ensuremath{\mathsf{negl}}\xspace}
\newcommand{\eps}{\ensuremath{\varepsilon}}
\newcommand{\ceil}[1]{\ensuremath{\left\lceil{#1}\right\rceil}\xspace}

\newcommand{\p}[1]{\ensuremath{^{(#1)}}\xspace}

\newcommand{\abs}[1]{\ensuremath{\left\vert{#1}\right\vert}\xspace}
\newcommand{\defeq}[0]{\ensuremath{{\;\vcentcolon=\;}}\xspace}
\newcommand{\eqdef}[0]{\ensuremath{{\;=\vcentcolon\;}}\xspace}
\newcommand{\getsr}[0]{\mathbin{\stackrel{\mbox{\,\tiny \$}}{\gets}}}
\newcommand{\band}[0]{\ensuremath{~\wedge~}\xspace}
\newcommand{\bor}[0]{\ensuremath{~\vee~}\xspace}
\newcommand{\union}[0]{\ensuremath{\cup}\xspace}

\newcommand{\GF}[1]{\ensuremath{\mathbb{GF}}\xspace}
\newcommand{\?}{\ensuremath{\stackrel{?}{=}}\xspace}

\newcommand{\ie}{\text{i.e.}\xspace}
\newcommand{\etal}{\text{et al.}\xspace}

\newcommand{\eg}{\text{e.g.}\xspace}

\newadjustboxenv{balgo}{max width={\linewidth}, innerenv={boxedminipage}{\linewidth}}
\newenvironment{boxedalgo}{\begin{center}\begin{balgo}}{\end{balgo}\end{center}}

\newcommand{\pred}[1]{\ensuremath{\mathsf{#1}}\xspace}

\newcommand{\Gen}{\ensuremath{\mathsf{Gen}}\xspace}
\newcommand{\Enc}{\ensuremath{\mathsf{Enc}}\xspace}
\newcommand{\Dec}{\ensuremath{\mathsf{Dec}}\xspace}
\newcommand{\sk}{\ensuremath{\mathsf{sk}}\xspace}
\newcommand{\pk}{\ensuremath{\mathsf{pk}}\xspace}
\DeclareMathOperator{\dist}{dist}
\newcommand{\HAM}{\ensuremath{\mathsf{HAM}}\xspace}
\newcommand{\ED}{\ensuremath{\mathsf{ED}}\xspace}

\newcommand{\Encin}{\ensuremath{\Enc_{in}}\xspace}
\newcommand{\Encout}{\ensuremath{\Enc_{out}}\xspace}
\newcommand{\Decin}{\ensuremath{\Dec_{in}}\xspace}
\newcommand{\Decout}{\ensuremath{\Dec_{out}}\xspace}
\newcommand{\Cin}{\ensuremath{C_{in}}\xspace}
\newcommand{\Cout}{\ensuremath{C_{out}}\xspace}

\newcommand{\LDC}{\ensuremath{\mathsf{LDC}}\xspace}
\newcommand{\NBS}{\ensuremath{\pred{NBS}}\xspace}

\newcommand{\Fool}{\pred{Fool}}
\newcommand{\Limit}{\pred{Limit}}
\newcommand{\Good}{\pred{Good}}

\newcommand{\Sign}{\pred{Sign}}
\newcommand{\Verify}{\pred{Ver}}
\newcommand{\Sigforge}{\pred{Sig}\text{-}\pred{forge}}
\newcommand{\Forge}{\pred{Frg}}

\newcommand{\Zplus}{\ensuremath{\bbZ^+}\xspace}
\newcommand{\rLDC}{\pred{rLDC}}

\newcommand{\crLDC}{\pred{crLDC}}

\newcommand{\SZ}{\pred{SZ}}
\newcommand{\rhosz}{\ensuremath{\rho_{\pred{sz}}}\xspace}
\newcommand{\betasz}{\ensuremath{\beta_{\pred{sz}}}\xspace}

\newcommand{\tildex}{\ensuremath{\tilde{x}}\xspace}
\newcommand{\tildey}{\ensuremath{\tilde{y}}\xspace}

\newcommand{\tsigma}{\ensuremath{\tilde{\sigma}}\xspace}

\newcommand{\tildej}{\ensuremath{\tilde{j}}\xspace}
\newcommand{\tildem}{\ensuremath{\tilde{m}}\xspace}

\newcommand{\tpk}{\ensuremath{\tilde{\pk}}\xspace}
\newcommand{\BlockDecode}{\pred{BlockDec}}
\DeclareMathOperator{\supp}{supp}
\newcommand{\Lim}{\pred{L}}
\newcommand{\F}{\pred{F}}

\newcommand{\betain}{\ensuremath{\beta_{in}}\xspace}
\newcommand{\rhoin}{\ensuremath{\rho_{in}}\xspace}

\newcommand{\EncHam}{\ensuremath{\Enc_{\pred{H},\lambda}}\xspace}
\newcommand{\DecHam}{\ensuremath{\Dec_{\pred{H},\lambda}}\xspace}
\newcommand{\Cins}{\ensuremath{\cC_{\pred{Ins}}}\xspace}
\newcommand{\Ins}{\pred{I}}
\newcommand{\EncIns}{\ensuremath{\Enc_{\Ins,\lambda}}\xspace}
\newcommand{\DecIns}{\ensuremath{\Dec_{\Ins,\lambda}}\xspace}
\newcommand{\bl}{\pred{bl}}

\renewcommand{\circ}{\Vert}

\title{Computationally Relaxed Locally Decodable Codes, Revisited\IEEEauthorrefmark{1}\thanks{\IEEEauthorrefmark{1}Full-version of the work with the same title published at ISIT 2023, available at \url{https://doi.org/10.1109/ISIT54713.2023.10206655}. \copyright 2023 IEEE. Personal use of this material is permitted. Permission from IEEE must be obtained for all other uses, in any current or future media, including reprinting/republishing this material for advertising or promotional purposes,creating new collective works, for resale or redistribution to servers or lists, or reuse of any copyrighted component of this work in other works.}}

\author{%
	\IEEEauthorblockN{Alexander R. Block\IEEEauthorrefmark{2} and Jeremiah Blocki\IEEEauthorrefmark{3}}
	\IEEEauthorblockA{\IEEEauthorrefmark{2}Georgetown University and University of Maryland, College Park. Email: \texttt{alexander.r.block@gmail.com}} \IEEEauthorblockA{\IEEEauthorrefmark{3}Purdue University. Email: \texttt{jblocki@purdue.edu}}
}

\begin{document}
\maketitle

% !TeX root = ../full-main.tex
\begin{abstract}
	We revisit computationally relaxed locally decodable codes (crLDCs) (Blocki \etal, Trans. Inf. Theory '21) and give two new constructions.
	Our first construction is a Hamming crLDC that is conceptually simpler than prior constructions, leveraging digital signature schemes and an appropriately chosen Hamming code.
	Our second construction is an extension of our Hamming crLDC to handle insertion-deletion (InsDel) errors, yielding an InsDel crLDC.
	This extension crucially relies on the noisy binary search techniques of Block \etal (FSTTCS '20) to handle InsDel errors.
	Both crLDC constructions have binary codeword alphabets, are resilient to a constant fraction of Hamming and InsDel errors, respectively, and under suitable parameter choices have poly-logarithmic locality and encoding length linear in the message length and polynomial in the security parameter.
	These parameters compare favorably to prior constructions in the poly-logarithmic locality regime.
\end{abstract}

\section{Introduction}\label{sec:intro}
Locally decodable codes (\LDC{}s) are error-correcting codes that admit super-efficient (\ie, poly-logarithmic time) recovery of individual symbols of an encoded message by querying only a few locations into a received word.
%Such codes are desirable in many applications in which recovering an entire encoded message is unnecessary or when polynomial time decoding is prohibitively expensive (\eg, working over enormous data sets). 
%Additionally, these codes have deep connections within other areas of research, such as private information retrieval \cite{..} and probabilistically checkable proofs \cite{..}.
For an alphabet $\Sigma$ and a (normalized) metric $\dist$, a pair of algorithms $\Enc \colon \Sigma^k \rightarrow \Sigma^K$ and $\Dec \colon [k] \rightarrow \Sigma$ (for $[k]\defeq \{1,\dotsc,k\}$) is a \emph{$(\ell,\rho,p)$-\LDC} if $\Dec$ is a randomized oracle algorithm such that for any message $x$ and any received word $y'$, if $\dist(\Enc(x), y') \leq \rho$ then for every $i$, $\Dec^{y'}(i)$ makes at most $\ell$ queries to $y'$ and outputs $x_i$ with probability at least $p$.
Here, $k$ and $K$ are the \emph{message} and \emph{block lengths}, respectively, $k/K$ is the \emph{rate}, $\ell$ is the \emph{locality}, $\rho$ is the \emph{error-rate}, and $p$ is the \emph{success probability}.

Studied extensively in the context of worst-case \emph{Hamming errors} \cite{STOC:KatTre00,JCSS:STV01,KerenidisW04,Yekhanin08,DGY11,Efremenko12,Yekhanin12,KopSar16,KMRS17} where $\dist$ is the normalized Hamming distance ($\HAM$), \emph{Hamming \LDC{}s}
%\LDC{}s have been extensively studied in the context of worst-case \emph{Hamming errors} \cite{STOC:KatTre00,JCSS:STV01,KerenidisW04,Yekhanin08,DGY11,Efremenko12,Yekhanin12,KopSar16,KMRS17}, where $\dist$ is the normalized Hamming distance $\HAM$. 
%The normalized Hamming distance between two strings $u,v \in \Sigma^n$ for any $n \in \bbN$ is the number of positions in which $u$ and $v$ differ, normalized by $n$; formally, $\HAM(u,v) \defeq \abs{\{i \colon u_i \neq v_i\}} / n$. % \in [0,1]$. % (\ie, $\dist = \HAM$ is the normalized Hamming distance) \alex{Talk about $\dist$ being a normalized metric, say in english and mathematically when talking about the constraints.} , 
%In this setting, any encoded message $y = \Enc(x)$ receives adversarial symbol replacements to become $y'$ subject to the constraint that $\HAM(y,y') \leq \rho$.
%where an encoded message $y = \Enc(x)$ receives adversarial symbol replacements to become $y'$, subject to the constraint that $\HAM(y,y') \leq \rho$.
%Such \LDC{}s are known as \emph{Hamming \LDC{}s} and
%\alex{language in next sentence it too strong; rephrase since it's more like ``we have no clue''. Language like ``seem to have irreconcilable tradeoffs'' is more appropriate, followed by quantifying over the best known constructions and lower bounds.}
%Known as \emph{Hamming \LDC{}s}, these codes 
seem to have irreconcilable trade-offs between the rate, error-rate, and locality. %, and it seems to be impossible to optimize all parameters simultaneously. 
For constant error-rate (the target of most applications), the best known constructions with constant $\ell \geq 3$ locality have super-polynomial rate \cite{Yekhanin08,DGY11,Efremenko12}, for $\ell = 2$ it is known that $K = \Theta(\exp(k))$ \cite{KerenidisW04}, and the best known constructions with constant rate have super-logarithmic (sub-polynomial) locality \cite{KMRS17}. 
Furthermore, the best known lower bounds for general Hamming \LDC{}s with constant error-rate and locality $\ell \geq 3$ are $K = \Omega(k^{\frac{\ell+1}{\ell-1}})$ \cite{Woodruff07}, and any locality $\ell=3$ linear Hamming \LDC has $K = \Omega(k^2/\log(k))$ \cite{Woodruff12}.
See surveys \cite{Yekhanin12,KopSar16} for more details.

%Various relaxations of \LDC{}s have been introduced to remedy these dramatic tradeoffs.
To remedy these dramatic trade-offs, 
%Of interest to us %in this work \
%is the notion of \emph{relaxed \LDC{}s} (\rLDC{}s).
Ben-Sasson \etal \cite{BGHSV06} introduced \emph{relaxed \LDC{}s} (\rLDC{}s).
Relaxed \LDC{}s are \LDC{}s that additionally allow the decoder to output a symbol $\bot \not\in \Sigma$, which signifies that the decoder does not know the correct value, under the following restrictions: the decoder (a) does not output $\bot$ ``too often''; and (b) never outputs $\bot$ when the queried codeword is uncorrupted.
This relaxation yields \LDC{}s with constant locality and block length $K = k^{1+\eps}$ for (small) constant $\eps > 0$.
Blocki \etal consider a further relaxation of \rLDC{}s known as \emph{computationally relaxed \LDC{}s} (\crLDC{}s) \cite{BGGZ21}: \rLDC{}s that are only resilient against adversarial channels that are computationally bounded (\ie, probabilistic polynomial time (PPT) channels).
This relaxation, inspired by the work of Lipton \cite{Lipton94}, %(and subsequent follow-up works \cite{..}) 
%which considers constructing Hamming codes in the presence of a computationally bounded adversarial channel.
%The relaxation to computationally bounded channels 
yields \crLDC{}s %, assuming the existence of collision-resistant hash functions, 
with constant rate, constant error-rate, and polylog locality. 

%\paragraph*{Insertion-Deletion Errors}\label{sec:intro-insdel}
%These prior results only consider codes which are resilient against worst-case Hamming errors.
%Though the study of various Hamming \LDC{}s remains an active and fruitful area of research, 
Recently, advances in coding theory have turned their focus to understanding and constructing codes which are resilient to \emph{insertion-deletion errors} (InsDel errors) \cite{Levenshtein_SPD66,LATIN:KiwLoeMat04,GurWan17,BraGurZba18,ICALP:HaeShaSud18,STOC:HaeSha18,SODA:CHLSW19,ICALP:CJLW19,GurLi19,STOC:HaeRubSha19,FOCS:Haeupler19,LiuTjuXin19,SODA:CGHL21,GurLi20,SODA:CheLi21,GurHaeSha21,HaeSha21,SimBru21,CJLW22},
%Such codes generalize Hamming codes to a broader range of errors.
%In this error model, 
where an adversarial channel inserts and deletes a bounded number of symbols into and from the encoded message.
Known as \emph{InsDel codes}, % generalize Hamming codes to the broader class of InsDel errors.
%For InsDel codes, 
the metric $\dist$ considered is the (normalized) edit distance $\ED$, defined as the minimum number of symbol insertions and deletions to transform a string $u$ into a string $v$, normalized by $2\max\{|u|,|v|\}$.
%and the metric $\dist$ considered for these codes is the (normalized) edit distance.
Only recently have efficient InsDel codes %(\ie, polynomial time encoding and decoding) 
with asymptotically optimal rate and error-rate been well-understood \cite{STOC:HaeSha18,FOCS:Haeupler19,STOC:HaeRubSha19,LiuTjuXin19,GurHaeSha21}.

The study of \LDC{}s resilient to InsDel errors (InsDel \LDC{}s) has been scarce, with only a handful of results to date. 
Introduced by Ostrovsky and Paskin-Cherniavsky \cite{OstPan15}, to the best of our knowledge, \emph{all} InsDel \LDC constructions follow \cite{OstPan15} and utilize a so-called ``Hamming-to-InsDel compiler'' \cite{BBGKZ20,EPRINT:CheLiZhe20,BloBlo21},
%To the best of our knowledge, InsDel \LDC constructions fall into two categories: the construction of Haeupler and Shahrasb \cite{..} which constructs explicit synchronization strings which can be locally decoded
%which utilizes synchronization strings to construct InsDel \LDC{}s, and the line of work constructing InsDel \LDC{}s via a ``Hamming-to-InsDel compiler'' \cite{..,..,..}.
which transforms any Hamming \LDC into an InsDel \LDC, increasing both the rate and error-rate by a constant factor and increasing the locality by a polylog factor.
%\begin{enumerate*}[label=(\arabic*)]
%	\item the rate and error-rate are increased by a constant factor; and
%	\item the locality is increased by a poly-logarithmic factor.
%\end{enumerate*}
For example, any locality-$3$ Hamming $\LDC$ with block length $K$ and error-rate $\rho$ can be compiled into a locality-$3\cdot \polylog(K)$ InsDel $\LDC$ with block length $\Theta(K)$ and error-rate $\Theta(\rho)$.
%Very recently, Blocki \etal \cite{FOCS:BCGLZZ21} established the first lower bounds for InsDel \LDC{}s, showing that constant locality InsDel \LDC{}s must have exponential block length.

\subsection{Overview of Results}\label{sec:overview}
In this work, we revisit \crLDC{}s with respect to both Hamming and InsDel errors. 
%To begin, a \emph{relaxed locally decodable code} (\rLDC) is defined analogously to a standard \LDC (as given in \cref{sec:intro}), except the decoder is additionally allowed to output the symbol $\bot \not\in \Sigma$ which signifies that the decoder does not know the correct answer. 
%It is also required that the decoder does not output $\bot$ ``too often''.
%\emph{Computationally relaxed locally decodable codes} (\crLDC{}s) are \rLDC{}s with the additional property that the adversarial channel introducing errors into codewords is assumed to be some probabilistic polynomial time (PPT) algorithm, rather than a computationally unbounded algorithm (\eg, as with \rLDC{}s and \LDC{}s).
%\crLDC{}s are similar to the classical notion of locally decodable codes with the following modifications:
%\begin{enumerate*}[label=(\arabic*)]
%	\item The decoder is additionally allowed to output the symbol $\bot \not\in \Sigma$, which signifies that the decoder does not know the correct answer; and
%	
%	\item The adversarial channel introducing errors into codewords is assumed to be some probabilistic polynomial time algorithm.
%\end{enumerate*}
We begin by defining \crLDC{}s.
\begin{definition}[Computationally Relaxed Locally Decodable Codes]\label{def:crldc}
	Let $\cC = \{ C_\lambda[ K, k, q_1, q_2 ] \}_{\lambda \in \bbN}$ be a code family with encoding algorithms $\{ \Enc_\lambda \colon \Sigma_1 \rightarrow \Sigma_2 \}_{\lambda \in \bbN}$ where $|\Sigma_i| = q_i$.
	We say $\cC$ is a \emph{$(\ell, \rho, p, \delta, \dist)$-computationally relaxed locally decodable code} (\crLDC) if there exists a family of randomized oracle decoding algorithms $\{ \Dec_\lambda \colon [k] \rightarrow \Sigma_1 \}_{\lambda \in \bbN}$ such that: %properties.
	\begin{enumerate}
		\item For all $\lambda \in \bbN$ and any $\tildey \in \Sigma_2^*$, $\Dec^{\tildey}_\lambda(i)$ makes at most $\ell$ queries to $\tildey$ for any $i \in [k]$;\label{item:rldc-query}
		
		\item For all $\lambda \in \bbN$ and any $x \in \Sigma_1^k$, we have
		\begin{align*}
			\Pr[ \Dec_\lambda^{\Enc_\lambda(x)}(i) = x_i ] = 1
		\end{align*} 
		for all $i \in [k]$;\label{item:rldc-correct}
		
		\item Define binary predicate $\Fool(\tildey, \rho, p, x, y, \lambda) = 1$ iff
		\begin{enumerate}
			\item $\dist(y,\tildey)\leq \rho$; and 
			
			\item $\exists i \in [k]$ such that 
			\begin{align*}
				\Pr[\Dec_\lambda^{\tildey}(i) \in \{x_i, \bot\}] < p,
			\end{align*} 
			where the probability is taken over $\Dec_\lambda$; 
		\end{enumerate}
		otherwise $\Fool(\tildey, \rho, p, x ,y, \lambda) = 0$.
		We require that for all PPT adversaries $\cA$ there exists a negligible function %\footnote{A function is \emph{negligible} if it is $o(x^{-c})$ for all constants $c>0$.} 
		$\eps_{\F}(\cdot)$ such that for all $\lambda \in \bbN$ and all $x \in \Sigma_1^k$, we have
		\begin{align*}
			\Pr[ \Fool(\cA(y), \rho, p, x, y, \lambda) = 1 ] \leq \eps_{\F}(\lambda),
		\end{align*}
		where the probability is taken over $\cA$ and $y = \Enc_\lambda(x)$.\label{item:rldc-fool}
		
		\item Define binary predicate $\Limit(\tildey, \rho, \delta, x, y, \lambda) = 1$ iff 
		\begin{enumerate}
			\item $\dist(y,y') \leq \rho$; and
			\item $|\Good(y')| < \delta \cdot k$, where 
			\begin{align*}
				\Good(y') \defeq \{ i \in [k] \colon \Pr[\Dec_\lambda^{y'}(i) = x_i] > 2/3 \}
			\end{align*}
			and the probability is taken over $\Dec_\lambda$; 
		\end{enumerate}
		otherwise $\Limit(\tildey, \rho, \delta, x, y, \lambda) = 0$.
		We require that for all adversaries PPT adversaries $\cA$ there exists a negligible function $\eps_{\Lim}(\cdot)$ such that for all $\lambda \in \bbN$ and all $x \in \Sigma_1^k$, we have $\Pr[\Limit(\cA(y), \rho, \delta, x, y, \lambda) = 1] \leq \eps_{\Lim}(\lambda)$,
%		\begin{align*}
%			\Pr[\Limit(\cA(y), \rho, \delta, x, y, \lambda) = 1] \leq \eps_{\Lim}(\lambda),
%		\end{align*}
		where the probability is taken over $\cA$ and $y = \Enc_\lambda(x)$.\label{item:rldc-limit}
	\end{enumerate}
	If $\dist$ is the normalized Hamming distance $\HAM$, we say the code is a \emph{Hamming \crLDC}; if $\dist$ is the normalized edit distance $\ED$, we say the code is a \emph{InsDel \crLDC}.
	Here, $\ell$ is the \emph{locality}, $\rho$ is the \emph{error-rate}, $p$ is the \emph{success probability}, and a function is \emph{negligible} if it is $o(x^{-c})$ for all constants $c>0$. %, and $k/K$ is the \emph{rate}.
	If $q_2 = 2$, we say that $\cC$ is a family of \emph{binary \crLDC{}s}, and if $q_1 = q_2$ we simply write $C_\lambda[K,k,q_1]$.
\end{definition}
\noindent %\cref{def:crldc} captures the notion of asymptotic security when interacting with arbitrary PPT adversaries; this differs from the standard definition of (relaxed) locally decodable codes, as we are concerned with worst-case errors with respect to these restricted adversaries.
\cref{def:crldc} closely follows the \crLDC definition of Blocki \etal \cite{BGGZ21} with a few modifications.
First, the constructions of \cite{BGGZ21} utilize a public random seed for a collision-resistant hash function, so their \crLDC definition is quantified over the randomness of the seed generation algorithm.
Our constructions do not require a public random seed so we omit this algorithm from our definition and instead quantify the security of our \crLDC over a code family $\{C_\lambda\}_{\lambda \in \bbN}$. 
This quantification also captures the notion of asymptotic security when interacting with PPT adversaries, which differs from standard (\pred{r})\LDC definitions that consider information-theoretic adversaries. % as we are only concerned with worst-case errors induced by PPT adversaries.
Moreover, \cite{BGGZ21} requires the public random seed to be generated in an honest (\ie, trusted) way, and our definition and constructions circumvent this requirement.
Second, we slightly strengthen the security definition by tweaking the predicate $\Fool$: in \cref{def:crldc}, the adversary wins if there \emph{exists} an index $i$ (not necessarily known by the adversary) such that the probability the decoder outputs correctly on input $i$ is less than $p$.
In contrast, \cite{BGGZ21} requires the adversary to output corrupt codeword $y'$ and a target index $i$ such that the probability the decoder outputs correctly on index $i$ is less than $p$.
Note that requiring \cref{def:crldc} to hold for $p = 2/3$, $\eps_{\F}(\lambda) = \eps_{\Lim}(\lambda) = 0$, and for all computationally unbounded adversaries $\cA$ results in the original \rLDC definition \cite{BGHSV06}.

Our first contribution is constructing a family of binary Hamming \crLDC{}s % $\cC_{\HAM} = \{ C_{\lambda}[K,k,2] \}$
satisfying \cref{def:crldc}.
%\textcolor{red}{Our construction is conceptually simpler than the Hamming \crLDC of Blocki \etal \cite{BGGZ21}} \alex{elaborate somehow}, does not require a trusted setup, and achieves (asymptotically) the same rate, error-rate, and locality.
Our construction borrows from code concatenation techniques \cite{concat}, which utilize an outer code $\Cout = (\Encout, \Decout)$ and an inner code $\Cin = (\Encin, \Decin)$ and encodes a message $x$ as follows:
\begin{enumerate}
	\item compute $y = \Encout(x)$;
	\item partition $y$ into some number $d$ of blocks $y\p{1}\circ \dotsc \circ y\p{d}$; 
	\item compute $Y\p{i} = \Encin(y\p{i})$ for all $i$; and
	\item output $Y = Y\p{1} \circ \dotsc \circ Y\p{d}$; here, $\circ$ denotes string concatenation.
\end{enumerate}
In our construction, we use the identity function as $\Cout$, utilize a suitable \emph{digital signature scheme} to sign each block $y\p{i}$, and use a classical Hamming code as $\Cin$. % to obtain our construction.
Briefly, a digital signature scheme with signatures of length $r(\cdot)$ is a tuple of PPT algorithms $\Pi = (\Gen,\Sign,\Verify)$ that satisfy the following properties: 
\begin{enumerate}
	\item $\Gen$ takes as input security parameter $\lambda \in \bbN$ (in unary) and outputs a key pair $(\pk,\sk)$, where $\pk$ is the \emph{public/verification key} and $\sk$ is the \emph{private/signing key};
	
	\item $\Sign$ takes as input a message $m$ of arbitrary length and the signing key $\sk$ and outputs a signature $\sigma \in \bin{r(\lambda)}$ of message $m$.
	
	\item $\Verify$ is deterministic and takes as input a message $m$, some signature $\sigma$, and a verification key $\pk$ outputs $1$ iff $\sigma$ is a valid signature of message $m$ and $0$ otherwise.
	
	\item For all PPT adversaries $\cA$, for $(\pk, \sk) \gets \Gen(1^\lambda)$, if $\cA$ is given $\pk$ as input and given oracle access to $\Sign_\sk(\cdot)$, then $\Pi$ is \emph{secure} if, except with negligible probability in $\lambda$, $\cA$ cannot output a pair $(\tildem, \tsigma)$ such that $\Verify_{\pk}(\tildem, \tsigma)=1$ and $\cA$ never queried $\Sign_\sk(\tildem)$.
\end{enumerate}
%A signature scheme $\Pi$ is said to be \emph{secure} if for all PPT adversaries $\cA$, for $(\pk, \sk) \gets \Gen(1^\lambda)$, if $\cA$ is given $\pk$ as input and given oracle access to $\Sign_\sk(\cdot)$, then, except with negligible probability in $\lambda$, $\cA$ cannot output a pair $(\tildem, \tsigma)$ such that $\Verify_{\pk}(\tildem, \tsigma)=1$ and $\cA$ never queried $\Sign_\sk(\tildem)$.
Given a secure digital signature scheme and any binary Hamming code, we obtain our first main result.

\begin{restatable}{theorem}{hamCRLDC}\label{thm:ham-crldc}
	Let $\Pi$ be a $r \defeq r(\lambda)$ length signature scheme.
	Let $\Cin$ be a binary Hamming code with rate $\betain$ and error-rate $\rhoin$.
	Then for every positive polynomial $k(\cdot)$ and constant $c \in (0,1/2)$, there exists a code family $\cC_{\pred{H}} \defeq \{C_{\pred{H},\lambda}[K,k(\lambda),2]\}_{\lambda \in \bbN}$ and function $\mu\defeq \mu(\lambda)$ such that $\cC_{\pred{H}}$ is a $(\ell, \rho, p, \delta)$-Hamming \crLDC with 
	\begin{itemize}
		\item $K = O((1/\betain)\max\{k (1 + \log(k)/r), r\})$, 
		\item $\ell = O((\mu/\betain)\cdot (r + \log(k)))$,
		\item $\rho = c \cdot \rhoin$, 
		\item $p = 1-\exp(-\mu (1/2-c)^2\big/2(1-c)) > 2/3$, and
		\item $\delta = 1/2$,
	\end{itemize}
	where $k \defeq k(\lambda)$.
\end{restatable}
\noindent Our code family $\cC_{\pred{H}}$ is constant rate whenever $\betain = \Theta(1)$ and $\Omega(\log(k(\lambda))) = r(\lambda) \leq k(\lambda)$.
Our construction allows for $r(\lambda) > k(\lambda)$, but this results in locality $\ell \geq K$, so it is more efficient to use a Hamming code with comparable rate and error-rate.
Any choice of $\mu$ satisfying $p > 2/3$ ensures that $\delta = 1/2$; % (hence why we state ``there exists'' rather than ``for all''). 
\eg, $\mu(\lambda) \defeq O(\log^{1+\epsilon}(\lambda))$ for constant $\epsilon > 0$ gives us polylog locality and success probability $1-\negl(\lambda)$, %yielding a \crLDC with high success probability at the cost of higher locality; here and throughout, we let 
where $\negl$ denotes some unspecified negligible function.

We can instantiate \cref{thm:ham-crldc} with a constant rate and error-rate binary Hamming code $\Cin$ (\eg, \cite{Justesen72}) and an appropriate signature scheme to achieve a constant rate and error-rate Hamming \crLDC with polylog locality.
Our construction shines when $r(\lambda) = \polylog(\lambda)$ and under standard idealized models there exist signature schemes with $r(\lambda)$ as small as $\Theta(\log^{1+\epsilon}(\lambda))$ for small constant $\epsilon > 0$ \cite{C:Schnorr89,EC:BloLee22}, assuming these schemes satisfy the following notion of concrete security: for security parameter $\lambda$, any adversary running in time $2^{\lambda/2}$ can violate the security of the scheme with probability at most $2^{-\lambda/2}$ for signatures of length $r(\lambda) = \lambda$.
Plugging in $\lambda' = \Theta(\log^{1+\epsilon}(\lambda))$, said schemes are secure against super-polynomial time adversaries with negligible security in $\lambda$, which implies they satisfy our definition of security for signature schemes.
%they are asymptotically secure against PPT adversaries.
Using such a scheme with a constant rate and error-rate Hamming code $\Cin$ and  $\mu(\lambda) \defeq O(\log^{1+\epsilon}(\lambda))$, we obtain the following corollary.

\begin{corollary}\label{cor:ham}
	Let $\Pi$ be a $r(\lambda) = \Theta(\log^{1+\epsilon}(\lambda))$ length signature scheme for constant $\epsilon > 0$.
	Then for all sufficiently large positive polynomials $k(\cdot)$, there exists code family $\{ C_{\pred{H},\lambda}[K, k(\lambda), 2] \}_{\lambda \in \bbN}$ that is a $(\ell, \rho, p, \delta)$-Hamming \crLDC with 
	\begin{itemize}
		\item $K = O(k)$, 
		\item $\ell = O( \log^{2(1+\epsilon)}(\lambda) )$, 
		\item $\rho = \Theta(1)$, 
		\item $p = 1-\negl(\lambda)$, and 
		\item $\delta = 1/2$,
	\end{itemize} 
	where $k \defeq k(\lambda)$. % and $\negl(\cdot)$ is an unspecified negligible function.
\end{corollary}

The parameters of \cref{cor:ham} are comparable to the Hamming \crLDC construction of \cite{BGGZ21}, which achieves $K = O(k)$, $ \ell = \polylog(k)$, $\rho = \Theta(1)$, $p = 1-\negl(\lambda)$, and $\delta = \Theta(1)$. 
Our construction is arguably conceptually simpler than that of \cite{BGGZ21}, which utilizes local expander graphs and collision-resistant hash functions (with a trusted setup), whereas our construction simply partitions, signs, and encodes. 
Moreover, our use of signatures does not require public key infrastructure as such schemes exist from one-way functions \cite{lamport}. %, whereas collision-resistant hash functions are a stronger assumption than one-way functions \cite{..}. %, whereas collision-resistant hash functions are (black-box) separated from one-way functions \cite{..}.

\subsubsection{Extension to InsDel Errors}
Our second contribution is extending the construction of \cref{thm:ham-crldc} to handle InsDel errors.
Prior constructions of InsDel $\LDC$s utilized a so-called ``Hamming-to-InsDel'' compiler \cite{OstPan15,BBGKZ20}. %, which also borrows from the notion of concatenation codes. %, using a suitable Hamming \LDC as the outer code $\Cout$ and a suitable InsDel code as the inner code (\ie, a non-local code).
%This new InsDel \LDC has asymptotically the same rate and error-rate as the underlying Hamming \LDC at the cost of a poly-logarithmic blow-up in the locality.
Key to this compiler is a \emph{noisy binary search} algorithm, which intuitively allows one to search an almost sorted list and find most entries with high probability. 
We use this algorithm to find blocks of codewords that are not ``too corrupt'', allowing us to handle more general InsDel errors.
We use the noisy binary search tools of Block \etal \cite{BBGKZ20} and the well-known Schulman-Zuckerman InsDel code \cite{SchZuc99} for $\Cin$ to extend \cref{thm:ham-crldc} to the InsDel setting.
%Using these tools and a suitable InsDel code for \Cin, we modify the family $\cC_{\pred{H}}$ to obtain a family of InsDel \crLDC{}s.
%We additionally modify the family $\cC_{\HAM}$ to make use of a suitable inner code $\Cin$ that is resilient to InsDel errors.
%We use the well-known Schulman-Zuckerman InsDel code \cite{SchZuc99} for $\Cin$, which has constant rate and error-rate, and has additional properties required by the noisy binary search tools of Block \etal \cite{BBGKZ20}.
Together with a secure digital signature scheme, we obtain our second main result.

\begin{restatable}{theorem}{mainthm}\label{thm:insdel-crldc}
	Let $\Pi$ be a $r \defeq r(\lambda)$ length signature scheme. 
	There exists a constant $c \in (0,1/2$) such that for every positive polynomial $k(\cdot)$ and constant $\rho^* \in (0,1/3)$, there exists a code family $\Cins \defeq \{ C_\lambda[K, k(\lambda), 2]\}_{\lambda \in \bbN}$ and a function $\mu \defeq \mu(\lambda)$ such that $\Cins$ is a $(\ell, \rho, p, \delta)$-InsDel $\crLDC$ with 
	\begin{itemize}
		\item $K = O(\max\{k(1+\log(k)/r),r\})$, %$n = \Theta(k \cdot [3 + \log(k) / r(\lambda)])$, 
		\item $\ell = O((\log^3(K)+\mu)\cdot(r + \log(k)))$, 
		\item $\rho = \Theta(1)$, 
		\item $p = 1-\rho^* - \exp(-\mu(1/2-c)^2/2(1-c)) > 2/3$, and 
		\item $\delta = 1-\Theta(\rho)$,
	\end{itemize} 
	where $k \defeq k(\lambda)$. %, where $c\geq 1$ is a constant that depends on the scheme $\Pi$.
\end{restatable}
\noindent As with \cref{thm:ham-crldc}, our family $\Cins$ is constant rate whenever $\Omega(\log(k(\lambda))) = r(\lambda) \leq k(\lambda)$, and additionally has the same downside whenever $r(\lambda) > k(\lambda)$, in which case it is more efficient to directly encode with an (asymptotically) optimal InsDel code (\eg, \cite{SchZuc99}).
We again choose $\mu$ such that $p = 1 - \negl(\lambda) > 2/3$; %, yielding high success probability at the cost of higher locality.
moreover, under the same set of assumptions on the underlying signature scheme as with our Hamming \crLDC (\eg, \cite{C:Schnorr89,EC:BloLee22}), for $\mu(\lambda) = \Theta(\log^{1+\epsilon}(\lambda))$ for small constant $\epsilon > 0$, we obtain the following corollary.

\begin{corollary}\label{cor:main}
	Let $\Pi$ be a $r(\lambda) = \Theta(\log^{1+\epsilon}(\lambda))$ length signature scheme for constant $\epsilon > 0$.
	Then for all sufficiently large positive polynomials $k(\cdot)$, there exists code family $\{C_{\Ins,\lambda}[K, k(\lambda), 2]\}_{\lambda \in \bbN}$ that is a $(\ell, \rho, p, \delta)$-InsDel \crLDC with $K = O(k)$, $\ell = O(\log^{4+\epsilon}(\lambda))$, $\rho = \Theta(1)$, $p = 1-\negl(\lambda)$, and $\delta = 1- \Theta(\rho)$, where $k \defeq k(\lambda)$. % and $\negl(\cdot)$ is an unspecified negligible function.
\end{corollary}
To the best of our knowledge, our InsDel $\crLDC$s are the first of their kind and compare favorably to the prior InsDel $\LDC$s %(\ie, non-relaxed and computationally unbounded errors) 
of Block \etal \cite{BBGKZ20} and are comparable to the private and resource-bounded \LDC{}s of Block and Blocki \cite{BloBlo21}.
\subsection{Related Work}\label{sec:prior}
Classical InsDel codes were initially studied in \cite{Levenshtein_SPD66}, inspiring a rich line of research into these codes; see surveys \cite{Sloane2002OnSC,SWAT:Mitzenmacher08,IEEE:MerBhaTar10} for more information.
Recently, $k$-deletion correcting codes with optimal rate were constructed, answering a long standing open question \cite{SimBru21}. %,ISIT:SimBru19,ISIT:SimGabBru20b,ISIT:SimGabBru20a}. 
Randomized codes with positive rate that are correct a large fraction of deletions are studied in \cite{LATIN:KiwLoeMat04,GurWan17}. %, and constant rate codes resilient to a constant fraction of InsDel errors with efficient encoding and decoding (\ie, polynomial time) were studied extensively in \cite{SchZuc99,STOC:HaeSha18,SODA:CHLSW19,GurLi19,BraGurZba18,SODA:CGHL21,SODA:CheLi21,GurHaeSha21,HaeSha21,CJLW22}.
Another line of work extends list decoding to InsDel codes \cite{ICALP:HaeShaSud18,LiuTjuXin19,GurHaeSha21}.
%List decodable codes are error-correcting codes that are resilient to a larger fraction of errors at the cost of outputting a small list of potential codewords , rather than a unique codeword.
Finally, \cite{STOC:HaeSha18} constructs explicit synchronization strings which can be ``locally decoded'' in the following sense: each index of the string is computable using symbols located at a small number of other locations in the string.
These synchronization strings are used to construct near linear time interactive coding schemes for InsDel errors.

\cite{Lipton94} initiated the study of codes resilient to errors introduced by computationally bounded channels. 
Several follow-up works adopt this channel model, yielding Hamming codes with better parameters than their classical counterparts \cite{TCC:MPSW05,GurSmi16,ShaSil21}.
It has been argued that any real-world communication channel can be reasonably modeled as a computationally bounded channel \cite{Lipton94,ITC:BloKulZho20}, %as all such channels have some sort of limitations on their computation, 
so one can reasonably expect error patterns encountered in nature to be modeled by some (possibly unknown) PPT algorithm.
This channel model has also been extended to the \LDC setting for both Hamming \cite{ICALP:OstPanSah07,C:HemOst08,HOSW11,ITC:BloKulZho20,BGGZ21} and, more recently, InsDel errors \cite{BloBlo21}.

%As discussed before, 
\cite{BGHSV06} introduced the notion of relaxed locally decodable codes.
%These codes admit local decoding algorithms with the additional property that the decoder is allowed to output a symbol $\bot$ which represents that the decoder does not know the correct value.
%This relaxation allows \cite{STOC:BGHSV04} to construct locally decodable codes with much better parameters than their classical counterparts; in particular, they achieve codes which are resilient to a constant fraction of Hamming errors, have constant locality, and have encoding length $k^{1+\epsilon}$ for small $\epsilon$.
In a follow-up work, \cite{GurRamRot20} introduced and construct \emph{relaxed locally correctable codes} (\pred{rLCC}) for Hamming errors: codes with local correction algorithms which can correct corrupt codeword symbols via querying a few locations into the received word.
Their construction has significantly better parameters than classical Hamming \pred{LCC}{}s, achieving constant locality, constant error-rate, and polynomial block length.
Furthermore, their \pred{rLCC} is also a \rLDC since their code is systematic. % (\ie, the message is a substring of the codeword).
Follow-up work continued to give improved \pred{rLCC} constructions \cite{ICALP:AsaShi21,FOCS:CohYan22,CGS22}
\cite{BGGZ21} studies Hamming \rLDC{}s/\pred{rLCC}{}s in the context of computationally bounded channels (\crLDC/\pred{crLCC}).
Our work directly adapts this model but for InsDel errors.
%The dual assumption of both a computationally bounded channel and a \rLDC allows \cite{BGGZ21} to construct \pred{crLCC}{}s and \crLDC{}s which achieve constant rate, constant error-rate, and poly-logarithmic locality. %, further improving on the results of Ben-Sasson \etal and Gur \etal \cite{BGHSV06,GurRamRot20}.

%As mentioned before, 
\cite{OstPan15} initiated the study of InsDel \LDC{}s.
They give a compiler which transforms any Hamming \LDC into an InsDel \LDC, asymptotically preserving the rate and error-rate of the underlying Hamming \LDC at the cost of a poly-logarithmic increase in the locality.
\cite{BBGKZ20} reproves this result with a conceptually simpler analysis using techniques borrowed from the study of a cryptographic object known as memory-hard functions \cite{ErdosGS75,CCS:AlwBloHar17,EC:AlwBloPie18,C:BloHol22}. %\cite{AC:ForLucWen14,STOC:AlwSer15,BiryukovDK16,C:AlwBlo16,EC:AlwBloPie17,EC:ACPRT17,TCC:AlwTac17,TCC:BloZho17,CCS:AlwBloHar17,EC:AlwBloPie18,C:BHKLXZ19,C:CheTes19,ITCS:ABZ20}.
\cite{EPRINT:CheLiZhe20} proposes the notion of Hamming/InsDel \LDC{}s with randomized encodings in various settings, including when the encoder and decoder share randomness or when the channel adds error patterns non-adaptively. 
In the InsDel case, \cite{EPRINT:CheLiZhe20} invokes the compiler of \cite{OstPan15} and obtain a code with block length $O(k)$ or $O(k \log(k))$ and $\polylog(k)$ locality.
Recently, \cite{BloBlo21} extends the compiler of \cite{BBGKZ20} to the private-key setting of \cite{ICALP:OstPanSah07}, where the encoder and decoder share a secret key unknown to the channel, and to the resource-bounded setting of \cite{ITC:BloKulZho20}, where the channel is assumed to be resource constrained in some way. % (\eg, the channel is a low-depth circuit).
%The aforementioned lower bounds for InsDel \LDC{}s due to Blocki \etal also extend to the private key \LDC setting; namely, for constant locality, even private key InsDel \LDC{}s require exponential block length \cite{FOCS:BCGLZZ21}.
While it is likely that applying the ``Hamming-to-InsDel'' compiler to the $\crLDC$ of \cite{BGGZ21} or our \crLDC{}s would yield an InsDel $\crLDC$, this has not been formally claimed or proven in prior work.
Finally, there has been recent progress in obtaining lower bounds for InsDel $\LDC$s.
\cite{FOCS:BCGLZZ21} proved that InsDel $\LDC$s with constant locality, even in the private-key setting, require exponential block length, and also show that linear 2-query InsDel $\LDC$s do not exist.
This makes it all the more surprising that a constant rate InsDel $\crLDC$ in the polylog locality regime exist.

%\alex{need to talk about lower bounds. put in intro and/or here.}
%\jeremiah{Stress that our results give a simplified construction for Hamming CRLDC}
%\jeremiah{For Insdel LDCs there are exponential lower bounds for any constant locality. This makes it more surprising that one can achieve constant rate, polylog locality in the insdel setting.}
%\jeremiah{It is likely that applying the compiler to CRLDC of BGGZ would yield a insdel CRLDC thought this has not been formally claimed or proven in prior work.}

%\input{full-version/tech.tex.old}
% !TeX root = ../full-main.tex
\section{Technical Overview}\label{sec:tech-ing}
The main technical ingredients for both our Hamming and InsDel \crLDC constructions are the use of a digital signature scheme $\Pi$ with $r$-length signatures along with a suitable inner code $\Cin$.
The encoding algorithms for both codes are nearly identical, with the main difference being the choice of $\Cin$. % (Hamming vs. InsDel). %, along with some padding in the InsDel case that is necessary for the noisy binary search tools we utilize.
The decoding algorithms are also similar: the InsDel decoder is a (non-trivial) modification of the Hamming decoder to handle InsDel errors using noisy binary search techniques.
%We begin by discussing our Hamming construction followed by our InsDel construction.

\subsubsection{Hamming \texorpdfstring{\crLDC}{crLDC} Construction}\label{sec:tech-hamming}
Let $\Cin$ be an appropriate Hamming code (\ie, non-local), and let $\Pi = (\Gen, \Sign, \Verify)$ be an $r$-length signature scheme.
%As mentioned in \cref{sec:intro}, our encoder borrows ideas from concatenation codes \cite{concat}, except we utilize $\Pi$ as our outer code.
%Concatenation codes consist of an outer code and an inner code, and (roughly) operate as follows: first, encode a message with the outer encoder; second, partition the outer code word in some way; third, encode each partition with the inner encoder; finally, output the concatenation each inner codeword.
%In place of an outer code, we utilize the digital signature scheme $\Pi$, and we use $\Cin$ as our inner code.

\paragraph*{The Hamming Encoder}
%We describe $\EncHam$ in this section.
We define a family of encoding algorithms $\{\EncHam\}_\lambda$. % in \cref{alg:ham-encoder}.
Let $\lambda \in \bbN$ be the security parameter.
For any message $x \in \bin{k}$, encoder \EncHam partitions $x$ into $d = \ceil{k/r(\lambda)}$ blocks $x = x\p{1} \circ \cdots \circ x\p{d}$, where $x\p{i} \in \bin{r(\lambda)}$ for all $i$ (padding with $0$ as necessary). %\footnote{For simplicity in this overview, assume $k/r(\lambda)$ is an integer.}
%Now in place of using an outer code to encode each block $x\p{i}$, we utilize our digital signature scheme. 
Each $x\p{i}$ is now signed using $\Pi$:  \EncHam generates key pair $(\pk,\sk) \gets \Gen(1^\lambda)$ and computes signature $\sigma\p{i} \gets \Sign_{\sk}(x\p{i}\circ i)$. 
Next, the block $x\p{i} \circ \sigma\p{i} \circ \pk \circ i$ is encoded using $\Cin$ to obtain codeword $c\p{i}$, where $\pk$ is the public key generated previously.
Finally, \EncHam outputs $C = c\p{1} \circ \cdots \circ c\p{d} \in \bin{K}$.
If $r(\lambda) \geq k$, only a single block is signed and encoded at the cost of locality $\geq K$, so 
%Note that whenever $r(\lambda) \geq k$, \EncHam would sign and encode a single block; however,  the locality of this codeword would be $\geq n$. % decoder will be larger than the final codeword length $n$ since there is only a single block to decode.
%At that point, 
it is more efficient to %either choose a larger $k$ or 
use a Hamming code with similar rate and error-rate rather than $\EncHam$. % (\ie, just use $\Cin$ without signing anything).
We give the formal encoding algorithm in \cref{alg:ham-encoder}.

% !TeX root = ../full-main.tex
%\begin{center}
	\begin{algorithm}
		\DontPrintSemicolon
		\caption{Hamming Encoder $\EncHam$}\label{alg:ham-encoder}
		\SetKwInOut{Input}{Input}
		\SetKwInOut{Output}{Output}
		\SetKwInOut{Hardcoded}{Hardcoded}
		\Input{A message $x \in \bin{k}$.}
		\Output{A codeword $C \in \bin{K}$.}
		\Hardcoded{Hamming code $\Cin$; $r(\cdot)$-length signature scheme $\Pi$; and $\lambda \in \bbN$ in unary.}
		
		Sample $(\pk, \sk) \gets \Gen(1^\lambda)$
		
		Set $d = \ceil{k/r(\lambda)}$
		
		Partition $x = x\p{1} \circ \cdots \circ x\p{d}$ where $x\p{j} \in \bin{r(\lambda)}$ for every $j \in [d]$ (padding last block as necessary)
		
		\ForEach{$j \in [d]$}{
			$\sigma\p{j} \gets \Sign_{\sk}(x\p{j}\circ j)$
			
			$C\p{j} =  \Encin(x\p{j} \circ \sigma\p{j} \circ \pk \circ j)$
%			\begin{gather}
%				\sigma\p{j} \gets \Sign_{\sk}(x\p{j}\circ j) \label{eq:ham-sign}\\
%				C\p{j} =  \Encin(x\p{j} \circ \sigma\p{j} \circ \pk \circ j)\label{eq:ham-enc}
%			\end{gather}
		}
		
		Define $C \defeq C\p{1} \circ \cdots \circ C\p{d} \in \bin{K}$
		
		\Return{C}
	\end{algorithm}
%\end{center}

%\input{ISIT-2023-camera-ready/alg-ham-encoder}

% !TeX root = ../full-main.tex
\paragraph*{Strawman Decoder}
Given $\EncHam$, there is a natural decoding algorithm that does not satisfy \cref{def:crldc}. %; however, this natural decoder cannot hope to satisfy \cref{def:crldc}.
The strawman decoder proceeds as follows.
Let $x \in \bin{k}$, $C = \EncHam(x) \in \bin{n}$, and let $\tC \in \bin{n}$ such that $\HAM(\tC, C) \leq \rho$.
Let $i \in [k]$ be the input given to the strawman decoder and let $\tC$ be its oracle.
Since the goal is to recover bit $x_i$ from string $\tC$, the strawman decoder first calculates index $j \in [d]$ such that bit $x_i$ resides in block $x\p{j}$.
Since $\tC$ only contains Hamming errors, the strawman decoder views its oracle $\tC$ as blocks $\tC\p{1} \circ \cdots \circ \tC\p{d}$ and recovers block $\tC\p{j}$.
The strawman decoder then runs the decoder of $\Cin$ with input $\tC\p{j}$ to obtain some string $\tildem\p{j}$ which can be viewed as some (potentially corrupt) string $\tildex\p{j} \circ \tsigma\p{j} \circ \tpk \circ \tildej$.
The strawman decoder then proceeds to use the signature scheme to verify the contents of this decoded message by checking if $\Verify_{\tpk}(\tildex\p{j} \circ j, \tsigma\p{j}) \? 1$.
If verification fails, then the decoder outputs $\bot$; otherwise, the decoder outputs $\tildex\p{j}_{i^*}$, where $i^*$ is the index of $x\p{j}$ that corresponds to bit $x_i$.

Notice that if $\tC = C$, then this strawman decoder outputs the correct bit $x_i$ with probability $1$, satisfying \cref{item:rldc-correct} of \cref{def:crldc}.
However, this strawman decoder can never satisfy \cref{item:rldc-fool} if we desire error-rate $\rho = \Theta(1)$.
Consider the following simple attack.
Let $\cA$ be a PPT adversary that operates as follows:
\begin{enumerate*}[label=(\arabic*)]
	\item Given codeword $C$, the adversary $\cA$ decodes block $C\p{1}$ to obtain $x\p{1} \circ \sigma\p{1} \circ \pk \circ 1$.
	\item $\cA$ then generates its own key pair $(\pk', \sk')$, a message $x' = 1-x\p{1}$, and computes $\sigma' = \Sign_{\sk'}(x' \circ 1)$.
	\item $\cA$ then computes $C' = \Encin( x' \circ \sigma' \circ \pk' \circ 1 )$ and outputs $\tC = C' \circ C\p{2} \circ \cdots \circ C\p{d}$.
\end{enumerate*}
Intuitively, this attack succeeds for two reasons.
The first reason is that corruption of $C\p{1}$ to $C'$ is a small fraction of the total amount of corruptions allotted to transform $C$ to $\tC$.
The second reason is that the strawman decoder relies on the public key $\pk'$to perform verification.
The key to preventing this attack is addressing the recovery of the public key.
Notice that if the decoder recovered the true public key $\pk$ used by $\EncHam$, then this attack fails since the verification procedure fails and the decoder outputs bot.
Thus we modify the strawman decoder to recover the true public key to obtain our final Hamming decoder.

\paragraph*{The Hamming Decoder}
We define a family of decoding algorithms $\{\DecHam\}_\lambda$. % in \cref{alg:ham-decoder}.
%With insight from the strawman decoder, we describe our actual decoder $\DecHam$.
%Recall that our goal is to recover the public key used by $\EncHam$.
%While we cannot recover this key with probability $1$, we can recover $\pk$ with sufficiently high probability.
%To do so, we utilize random sampling and with majority vote.
Let $\lambda \in \bbN$, $\mu \in \bbN$ be a parameter of our choice, $x \in \bin{k}$, $C = \EncHam(x)$, and $\tC \gets \cA(C)$ such that $\HAM(C, \tC) \leq \rho$, where $\cA$ is a PPT adversary.
On input $i \in [k]$ and given oracle access to $\tC$, the decoder $\DecHam$ tries to recover $x_i$ via a two-step process.
First, $\DecHam$ tries to recover the true public key $\pk$.
It begins by uniformly sampling $j_1, \dotsc, j_\mu \getsr [d]$.
Parsing $\tC$ as $\tC\p{1}\circ \cdots \circ \tC\p{d}$, for each $\kappa \in [\mu]$ \DecHam (1) recovers some $\tildem\p{j_\kappa} \gets \Decin( \tC\p{j_\kappa} )$; (2) parses $\tildem\p{j_\kappa}$ as $\tildex\p{j_\kappa} \circ \tsigma\p{j_\kappa} \circ \tpk\p{j_\kappa} \circ \tildej$; and (3) recovers $\tpk\p{j_\kappa}$.
%\begin{enumerate*}[label=(\arabic*)]
%	\item recovers some $\tildem\p{j_\kappa} \gets \Decin( \tC\p{j_\kappa} )$;
%	\item parses $\tildem\p{j_\kappa}$ as $\tildex\p{j_\kappa} \circ \tsigma\p{j_\kappa} \circ \tpk\p{j_\kappa} \circ \tildej$; and
%	\item recovers $\tpk\p{j_\kappa}$.
%\end{enumerate*}
The decoder then sets $\pk^* = \pred{majority}(\tpk\p{j_1},\dotsc, \tpk\p{j_\mu})$.
Second, \DecHam computes $j$ such that $x_i$ lies in $x\p{j}$, computes $\tildem\p{j} \gets \Decin(\tC\p{j})$, parses it as $\tildex\p{j} \circ \tsigma\p{j} \circ \tpk\p{j} \circ \tildej$, then checks if $\Verify_{\pk^*}(\tildex\p{j} \circ j, \tsigma\p{j}) =1$, outputting $\tildex\p{j}_{i^*}$, where $i^*$ is the index of $x\p{j}$ that corresponds to $x_i$ if true; else $\DecHam$ outputs $\bot$.
Here it is crucial for \DecHam to use its computed value $j$, otherwise it is possible for an adversary to swap two blocks $C\p{j_1}$ and $C\p{j_2}$ where $x\p{j_1} \neq x\p{j_2}$, violating \cref{item:rldc-fool}. % of \cref{def:crldc}.
We give the formal decoding algorithm in \cref{alg:ham-decoder}.

% !TeX root = ../full-main.tex
%\begin{center}
	\begin{algorithm}
		\DontPrintSemicolon
		\caption{Hamming Decoder $\DecHam$}\label{alg:ham-decoder}
		\SetKwInOut{Input}{Input}
		\SetKwInOut{Output}{Output}
		\SetKwInOut{Oracle}{Oracle}
		\SetKwInOut{Hardcoded}{Hardcoded}
		\Input{Index $i \in [k]$.}
		\Oracle{Bitstring $\tC \in \bin{K}$.}
		\Output{A symbol $\tildex \in \zo$ or $\bot$.}
		\Hardcoded{Hamming code $\Cin$; $r(\cdot)$-length signature scheme $\Pi$; $\lambda \in \bbN$ in unary; and $\mu \in \bbN$.}
		
		Set $d = \ceil{k/r(\lambda)}$ and $\bl = n/d$ %\tcc*{Compute number of blocks $d$ and block length $\bl$.}
		
		Sample $j_1,\dotsc, j_\mu \getsr [d]$\label{line:decham-sample}
		
		Initialize $\pk^*, \pk_1,\dotsc, \pk_\mu = 0$
		
		\ForEach(\tcp*[f]{Recover \pk}){$\kappa \in [\mu]$}{
			$\tildem\p{j_\kappa} \gets \Decin( \tC[ (j_\kappa-1)\cdot \bl +1, j_\kappa \cdot \bl ] )$\label{line:decham-pk-decode}
			
			Parse $\tildem\p{j_\kappa}$ as $\tildex\p{j_\kappa}\circ \tsigma\p{j_\kappa} \circ \tpk\p{j_\kappa}\circ \tildej$\label{line:decham-pk-parse}
			
			Set $\pk_\kappa = \tpk\p{j_\kappa}$\label{line:decham-pk-set}
		}
		
		Set $\pk^* = \pred{majority}(\pk_1,\dotsc, \pk_\mu)$\label{line:decham-pk-majority}
		
		Compute $j \in [d]$ such that $(j-1) \cdot r(\lambda) < i \leq j\cdot r(\lambda)$
		
		$\tildem\p{j} \gets \Decin( \tC[(j-1)\cdot\bl + 1, j\cdot \bl])$\label{line:decham-decode} %\tcc*{Recover block where $x_i$ should reside.}
		
		Parse $\tildem\p{j}$ as $\tildex\p{j} \circ \tsigma\p{j} \circ \tpk\p{j} \circ \tildej$
		
		\lIf{$\Verify_{\pk^*}(\tildex\p{j}\circ j, \tsigma\p{j}) = 0$\label{line:decham-verify}}{\Return{$\bot$}}
		
		\Return{$\tildex\p{j}_{i^*}$ for $i^* = i-(j-1)\cdot r(\lambda)$.}
	\end{algorithm}
%\end{center}

%\paragraph*{Hamming \texorpdfstring{\crLDC}{crLDC} Proof Overview}
\paragraph*{\texorpdfstring{\cref{thm:ham-crldc}}{Theorem 1} Proof Overview}
We give a high-level overview of the proof of \cref{thm:ham-crldc}; full details can be found in \cref{sec:ham-proof}.
The main technical challenge of proving \cref{thm:ham-crldc} is showing that $\{(\EncHam, \DecHam)\}_\lambda$ satisfies \cref{item:rldc-fool,item:rldc-limit} of \cref{def:crldc}.
%We give a high level overview of the proof. % and present the formal proof in \cref{sec:ham-proof}.
Towards \cref{item:rldc-fool}, for any $x \in \bin{k}$, PPT adversary $\cA$, and $i \in [k]$, we analyze the probability that $\DecHam^{\tC}(i) \in \{x_i, \bot\}$ for $\tC \gets \cA(\EncHam(x))$ such that $\HAM(\tC, \EncHam(x)) \leq \rho$.
If $\DecHam^{\tC}(i) = x_i$, then $\tildex\p{j}_{i^*} = x_i$ and $\Verify_{\pk^*}(\tildex\p{j}\circ j, \tsigma\p{j}) = 1$.
Conditioning on $(\tildex\p{j}\circ j,\tsigma\p{j})$ not breaking the security of $\Pi$, successful verification implies $\tildex\p{j} = x\p{j}$ and $\tsigma\p{j} = \sigma\p{j}$, and verification succeeds whenever $\pk^* = \pk$.
By Chernoff, we can ensure that $\pk^* = \pk$ with high probability (depending on $\mu$) as long as more than half of the blocks $\tC\p{j}$ have less than $\rhoin$-fraction of Hamming errors, which is achieved by setting $\rho = c \rhoin$ for any $c \in (0,1/2)$. % ensure this is the case.
%Here we conditioned on $(\tildex\p{j} \circ j,\tsigma\p{i})$ not violating the security of $\Pi$.
Now by definition, $\tsigma\p{i}$ does not break security of $\Pi$ with probability at least $1-\eps_{\Pi}(\lambda)$ (\ie, either it is a correct signature or verification fails), where $\eps_{\Pi}(\lambda)$ is a negligible function for the security of $\Pi$.
As index $i$ was arbitrary here, we establish \cref{item:rldc-fool} via a union bound for $p = 1- \exp(-\mu(1/2-c)^2\big/2(1-c))$ and $\eps_{\F}(\lambda) = k \cdot \eps_{\Pi}(\lambda)$, where $\eps_{\F}(\lambda)$ is negligible since $k=\poly(\lambda)$.

Towards \cref{item:rldc-limit}, define $\cJ \defeq \{ j \colon \HAM(\tC\p{j}, C\p{j}) \leq \rhoin \}$. 
Then $\rho = c \rhoin$ for $c \in (0,1/2)$ implies $|\cJ| \geq d/2$.
Moreover, $\Decin(\tC\p{j}) = x\p{j} \circ \sigma\p{j} \circ \pk \circ j$ for any $j \in \cJ$. %; \ie, we can recover the correct message, signature, public key, and index.
Now for any $i \in [k]$, if $x_i$ lies in $x\p{j}$ for $j \in \cJ$, the probability $\DecHam$ outputs $x_i$ equals the probability that $\DecHam$ correctly recovers $\pk^* = \pk$.
We choose $\mu$ to ensure $\Pr[\pk^* = \pk] > 2/3$.
This along with $|\cJ| \geq d/2$ implies that $|\Good(\tC)| \geq k/2$ and $\delta = 1/2$.
By our choice of $\rho$, $|\cJ| \geq d/2$ and $|\Good(\tC)| \geq k/2$ holds \emph{for any} $\tC$ such that $\HAM(C,\tC) \leq \rho$; %, beyond those given by PPT adversary $\cA$.
Thus \cref{item:rldc-limit} holds with $\delta = 1/2$ and $\eps_{\Lim}(\lambda) \defeq 0$.
%See \cite{BloBlo23} for the complete proof.

\subsubsection{InsDel \texorpdfstring{\crLDC}{crLDC} Construction}\label{sec:tech-insdel-crldc}
%Like our Hamming code construction for \cref{thm:ham-crldc}, the main technical ingredients of \cref{thm:insdel-crldc} are a digital signature scheme and a suitable inner code.
Let $\Cin$ be the Schulman-Zuckerman InsDel code ($\SZ$ code) \cite{SchZuc99} and let $\Pi=(\Gen,\Sign,\Verify)$ be an $r$-length signature scheme.
%Our construction additionally requires the use of the noisy binary search algorithm due to Block \etal \cite{BBGKZ20}.
%In particular, simply replacing $\Cin$ in $\EncHam$ of \cref{alg:ham-encoder} with the \SZ code and using $\DecHam$ does not yield an interesting InsDel $\crLDC$.

\paragraph*{Challenges to Decoding InsDel Errors}
InsDel errors allow an adversary to insert symbols into and delete symbols from codewords, which introduces challenges that do not arise with Hamming errors. 
One may hope to simply use our family $\{(\EncHam,\DecHam)\}_\lambda$ with $\Cin = \SZ$ to achieve \cref{thm:insdel-crldc}; however this yields a trivial InsDel \crLDC.
Let $\EncHam'$ be identical to $\EncHam$ except we use the $\SZ$ code as the code $\Cin$.
For any $x$ and $C = \EncHam'(x)$, there is a simple attack to ensure that $\DecHam$ \emph{always} outputs $\bot$: 
the adversary simply transforms $C = C\p{1} \circ \dotsc \circ C\p{d}$ into $\tC = C_1\p{d} \circ C\p{1} \circ \cdots \circ C\p{d-1} \circ C_0\p{d}$, where $C_0\p{d}, C_1\p{d}$ are the first and second halves of $C\p{d}$, respectively.
%Now recovery of any block $j$ and the public key $\pk$ is impossible.
This implies that $\{(\EncHam',\DecHam)\}_\lambda$ is an InsDel $\crLDC$ with $\delta = 0$; \ie, it always outputs $\bot$ given a corrupt codeword. % and achieves the same $\delta$. 
However, we can handle this and more general attacks by leveraging the noisy binary search techniques of Block \etal \cite{BBGKZ20}. 
%We first describe our InsDel encoder.

\paragraph*{Noisy Binary Search Overview}
%The noisy binary search algorithm of Block \etal \cite{BBGKZ20} allows us to address the above challenges.
%To begin, our InsDel encoder is nearly identical to $\EncHam$ with the following changes:
%\begin{enumerate*}[label=(\arabic*)]
%	\item we use the $\SZ$ code as our inner code $\Cin$; and
%	\item we additionally pad each block encoded by the $\SZ$ code with a suitable number of $0$s before and after each codeword, before concatenating all of them together to yield the final codeword.
%\end{enumerate*}
%To differentiate from $\EncHam$, we let $c\p{j} = \SZ.\Enc(x\p{j}\circ \sigma\p{j}\circ \pk \circ j)$ denote the encoded signed message block and let $C\p{j}$ denote $c\p{j}$ padded with an appropriate number of $0$s before and after the codeword.
%This padding is necessary for the noisy binary search algorithm.
To understand the noisy binary search algorithm \NBS and its guarantees, we require the notion of \emph{$\gamma$-goodness}.
For $x,y \in \bin{*}$, we say that $y$ is \emph{$\gamma$-good with respect to $x$} if $\ED(x,y) \leq \gamma$.
The notion of $\gamma$-goodness (albeit under different formal definitions) has been useful in the design and analysis of depth-robust graphs, a combinatorial object used extensively in the study of memory-hard functions \cite{CCS:AlwBloHar17,EC:AlwBloPie18,ErdosGS75}, and it is essential to the success of \NBS.
Intuitively, for a fixed ``correct'' ordered list of strings $A = (a_1,\dotsc, a_n)$, each of length $\kappa$, and some other list of strings $B = (b_1,\dotsc, b_{n'})$, the algorithm \NBS finds any string $b_j$ that is $\gamma$-good with respect to the string $a_j$ for $j \in [n]$, except with negligible probability. 
In our context, each $b_j$ corresponds to blocks in the (possibly corrupt) codeword.
Given a tolerance parameter $\rho^* \in (0,1/2)$, the \NBS algorithm on input $j \in [n]$ outputs $b_j$ for at least $(1-\rho^*)$-fraction of the $\gamma$-good indices $j$, except with negligible probability.
Moreover, \NBS runs in time $\kappa \cdot \polylog(n')$, which is only possible by allowing \NBS to fail on a small fraction of $\gamma$-good indices, else the algorithm requires $\Omega(\kappa n')$ time.
%Searching for blocks $\tildec\p{j}$, \NBS utilizes a so-called block decoding algorithm \BlockDecode to find $\tildec\p{j}$ within corrupt codeword $\tC$. 
%For input $i \in [|\tC|]$, if $i$ falls within a ball around a $\gamma$-good block $\tildec\p{j}$ in $\tC$, then \BlockDecode outputs (the decoding of) $\tildec\p{j}$ except with probability at most $\gamma$.
%Our construction leverages both the \BlockDecode and \NBS algorithms along with a suitable digital signature scheme $\Pi$.
%Given these tools, we describe our modified encoder and decoder.

Suppose that $\tC \in \bin{n'}$ for some $n'$ is a corrupt codeword (from an appropriate encoding algorithm) and let $i \in [k]$. 
We use the \NBS algorithm to search $\tC$ for some (possibly corrupt) block $\tildem\p{j}$ which contains the desired symbol $x_i$.
So long as $\tC$ and $\tildem\p{j}$ are ``not too corrupt'', then \NBS outputs $\tildem\p{j}$ with high probability.
For searching, \NBS utilizes a block decoding algorithm \BlockDecode to find $\tildem\p{j}$ within $\tC$ with the following guarantee: for input $i$, if $i$ is within a (small) ball around a $\gamma$-good block $\tildem\p{j}$, then \BlockDecode outputs $\tildem\p{j}$ with probability at least $1-\gamma$.
Assuming $\tildem\p{j}$ is not too corrupt, we can parse it as $\tildex\p{j} \circ \tsigma\p{j} \circ \tpk \circ \tildej$ and use $\Verify$ to ensure that $\tildex\p{j}$ is correct. 
Note that both \NBS and \BlockDecode can fail and output $\bot$, which we take into consideration for our decoder.

\paragraph*{The InsDel Encoder}
We define a family of encoding algorithms $\{\EncIns\}_\lambda$. 
Let $\lambda \in \bbN$ be the security parameter and let $\alpha$ be a constant specified by the \NBS algorithm \cite{BBGKZ20}.
For any message $x \in \bin{k}$, encoder $\EncIns$ behaves identically to $\EncHam$ by partitioning $x$ into $d=\ceil{x/r(\lambda)}$ blocks, generating $(\pk,\sk) \gets \Gen(1^\lambda)$, computing $\sigma\p{j} \gets \Sign_{\sk}(x\p{j}\circ j)$, and computing $c\p{j} = \SZ.\Enc(x\p{j}\circ \sigma\p{j}\circ \pk \circ j)$ for every $j$. 
Next, the encoder computes buffered codewords $C\p{j} = 0^{\alpha r(\lambda)} \circ c\p{j} \circ 0^{\alpha r(\lambda)}$ for every $j$, where $0^{\alpha r(\lambda)}$ is a all-zero vector of length $\alpha r(\lambda)$ and ensure the success of the \NBS and \BlockDecode algorithms.
Finally, $\EncIns$ outputs $C = C\p{1} \circ \cdots \circ C\p{d}$.
Again, if $r(\lambda) \geq k$, it is more efficient to simply to encode $x$ using the \SZ code. % (or any other InsDel code with comparable rate and error-rate). 
We give the formal encoding algorithm in \cref{alg:encoder}; any differences between this encoder and the Hamming encoder (\cref{alg:ham-encoder}) are highlighted in blue and with an inline comment.

% !TeX root = ../full-main.tex
%\begin{center}
	\begin{algorithm}
		\DontPrintSemicolon
		\caption{InsDel Encoder $\EncIns$}\label{alg:encoder}
		\SetKwInOut{Input}{Input}
		\SetKwInOut{Output}{Output}
		\SetKwInOut{Hardcoded}{Hardcoded}
		\Input{A message $x \in \bin{k}$.}
		\Output{A codeword $C \in \bin{K}$.}
		\Hardcoded{The $\SZ$ InsDel code; $r(\cdot)$-length signature scheme $\Pi$; $\alpha \in \bbN$; and $\lambda \in \bbN$ in unary.}
		
		Sample $(\pk, \sk) \gets \Gen(1^\lambda)$\label{line:enc-key-sample}
		
		Set $d = \ceil{k / r(\lambda)}$
		
		Partition $x=x\p{1} \circ \cdots \circ x\p{d}$ where $x\p{j} \in \bin{r(\lambda)}$ for every $j \in [d]$ (padding last block as necessary)\label{line:enc-block} 
		
		\ForEach(\tcp*[f]{$\EncHam$ Diff}){$j \in [d]$}{
			$\sigma\p{j} \gets \Sign_{\sk}(x\p{j} \circ j)$
			
			${\color{blue}c\p{j} = \SZ.\Enc(x\p{j} \circ \sigma\p{j} \circ \pk \circ j)}$
			
			${\color{blue}C\p{j} = (0^{\alpha \cdot \pred{ml}} \circ c\p{j} \circ 0^{\alpha \cdot \pred{ml}})}$; ${\color{blue}\pred{ml} = \abs{x\p{j} \circ \sigma\p{j} \circ \pk \circ j}}$
%			\begin{gather}
%				\sigma\p{j} \gets \Sign_{\sk}(x\p{j} \circ j)\label{eq:sign}\\
%				{\color{blue}c\p{j} = \SZ.\Enc(x\p{j} \circ \sigma\p{j} \circ \pk \circ j)\label{eq:sz-enc}}\\
%				{\color{blue}C\p{j} = (0^{\alpha \cdot \pred{ml}} \circ c\p{j} \circ 0^{\alpha \cdot \pred{ml}})\label{eq:ez-buff}},
%			\end{gather}
			%\tcc*{$\pred{ml}$ is the same for every $j$.}
		}
		
		Define $C \defeq C\p{1} \circ \cdots \circ C\p{d} \in \bin{K}$.
		
		\Return{C}
	\end{algorithm}
%\end{center}

%the encoder $\EncIns$ first samples a public/private key pair $(\pk, \sk) \gets \Gen(1^\lambda)$.
%Next, the message $x$ is partitioned into $d = \ceil{k / r(\lambda)}$ blocks $x\p{1} \circ \cdots \circ x\p{d}$, each of size $r(\lambda)$ bits.
%For every block $j \in \{1,\dotsc, d\}$, the signature $\sigma\p{j} \gets \Sign_\sk(x\p{j} \circ j)$ is computed, and a new block $m\p{j} \defeq x\p{j} \circ \sigma\p{j} \circ \pk \circ j$ is formed.
%This block is then encoded as $c\p{j} = \SZ.\Enc_\lambda(m\p{j})$, and computes a zero-buffered block $C\p{j} = (0^{\alpha \cdot r(\lambda)} \circ c\p{j} \circ 0^{\alpha \cdot r(\lambda)})$, where $\alpha$ is a suitably chosen constant \cite{BBGKZ20}. % (see \cref{lem:nbs}).
%The final codeword $C$ is the concatenation of all buffered blocks; \ie, $C \defeq C\p{1} \circ \cdots \circ C\p{d}$.
%As with our Hamming encoder, whenever $r(\lambda) \geq k$, it is more efficient to just encode $x$ with the $\SZ$ code.

%We now formally present our encoding and decoding algorithms.

\paragraph*{The InsDel Decoder}
We define a family of decoding algorithms $\{\DecIns\}_\lambda$.
Let $\lambda \in \bbN$, $\mu \in \bbN$ be a parameter of our choice, $x \in \bin{k}$, $C = \EncIns(x)$, and $\tC \gets \cA(C)$ such that $\ED(C, \tC) \leq \rho$, where $\cA$ is a PPT adversary and $\tC \in \bin{K'}$ for some $K'$.
Then on input $i \in [k]$ and given oracle access to $\tC$, the decoder $\DecIns$ tries to recover $x_i$ via the same two-step process as $\DecHam$:
%Given oracle access to $\tC$ and an index $i \in [k]$ as input, the decoder attempts to recover bit $x_i$ of the original message $x$.
%Calling back to the encoding algorithm, bit $x_i$ resides in block $x\p{j}$ for $j$ satisfying $(j-1) \cdot r(\lambda) < i \leq j\cdot r(\lambda)$ (\cref{line:enc-block} of \cref{alg:encoder}).
%Thus recovering $x_i$ is done by recovering block $x\p{j}$ from corrupt codeword $\tC$.
%Assuming $\DecIns$ can recover $x\p{j}$, the decoder simply outputs $x_i = x\p{j}_{i^*}$ for $i^* = i - (j-1) \cdot r(\lambda)$; otherwise, our goal will be for $\DecIns$ to output $\bot$.
%Same line as the Hamming decoder, our insertion-deletion decoder proceeds in two steps: 
first, recover the public key $\pk$; and second, find block $j$ that is supposed to contain $x_i$ and use the recovered \pk to verify its integrity. %recover some (possibly corrupt) block $\tildem\p{j} = (\tildex\p{j} \circ \tsigma\p{j} \circ \tpk\p{j} \circ \tildej)$.
%The block $\tildem$ can then be parsed and verified using $\Verify_\pk$; that is, the decoder runs $\Verify_\pk( \tildex \circ j, \tsigma )$ and outputs $\tildex_{i^*}$ if $\Verify_\pk( \tildex \circ j, \tsigma ) = 1$ and outputs $\bot$ otherwise.
%Note here that
%\begin{enumerate*}[label=(\arabic*)]
%	\item the public key $\pk$ recovered in the first step is not necessarily equal to string $\tpk$ parsed from $\tildem$ in the second step; and
%	\item $\pk$ and desired block index $j$ are used to verify the string $\tildex$ rather than the parsed public key $\tpk$ and parsed index $\tildej$.
%\end{enumerate*}
Recovery of \pk is done similarly to \DecHam, except we leverage \BlockDecode to find blocks with potential public keys by first sampling $i_1,\dotsc, i_\mu \getsr [n]$ uniformly at random, then obtaining $\tildem\p{j_\kappa}\gets\BlockDecode(i_\kappa)$ for each $\kappa \in [\mu]$, where $j_\kappa \in [d]$.
Intuitively, in the InsDel setting we need to search for each block $j_\kappa$ whereas in the Hamming setting we knew exactly where each block was located.
If $\tildem\p{j_\kappa} = \bot$, then we set $\pk_{\kappa} = \bot$; else, we parse $\tildem\p{j_\kappa}$ as $\tildex\p{j_\kappa}\circ \tsigma\p{j_\kappa} \circ \tpk\p{j_\kappa} \circ \tildej$ and set $\pk_\kappa = \tpk\p{j_\kappa}$.
Finally, we let $\pk^* = \pred{majority}(\pk_1,\dotsc, \pk_\kappa)$.
Next, $\DecIns$ computes $j$ such that $x_i$ lies in $x\p{j}$ and %.
%Then $\DecIns$ 
obtains $\tildem\p{j} \gets \NBS(j)$.
If either $\tildem\p{j}$ or $\pk^*$ are $\bot$, then $\DecIns$ aborts and outputs $\bot$.
Otherwise, $\tildem\p{j}$ is parsed as $\tildex\p{j}\circ \tsigma\p{j} \circ \tpk\p{j} \circ \tildej$ and then $\DecIns$ checks if $\Verify_{\pk^*}(\tildex\p{j} \circ j, \tsigma\p{j}) = 1$.
If not, $\DecIns$ outputs $\bot$; otherwise, $\DecIns$ outputs $\tildex\p{j}_{i^*}$, where $x_i = x_{i^*}\p{j}$.
We give the formal decoding algorithm in \cref{alg:decoder}; any differences between this encoder and the Hamming decoder (\cref{alg:ham-decoder}) are highlighted in blue and with an inline comment.

% !TeX root = ../full-main.tex
%\begin{center}
	\begin{algorithm}
		\DontPrintSemicolon
		\caption{InsDel Decoder \DecIns}\label{alg:decoder}
		\SetKwInOut{Input}{Input}
		\SetKwInOut{Output}{Output}
		\SetKwInOut{Hardcoded}{Hardcoded}
		\SetKwInOut{Oracle}{Oracle}
		\SetKw{continue}{continue}
		\Input{An index $i \in [k]$.}
		\Oracle{Bitstring $\tC \in \bin{K'}$ for some $K'$.}
		\Output{A symbol $\tildex \in \zo$ or $\bot$.}
		\Hardcoded{$\BlockDecode$ and $\NBS$; $r(\cdot)$-length signature scheme $\Pi$; block length $K$; $\lambda \in \bbN$ in unary; and $\mu \in \bbN$.}
		
		Set $d = \ceil{k / r(\lambda)}$
		
		\textcolor{blue}{Sample $i_1, \dotsc, i_\mu \getsr [n]$} \tcp*{$\DecHam$ Diff}
		
		Initialize $\pk^*, \pk_1, \dotsc, \pk_\mu = 0$
		
		\ForEach(\tcp*[f]{$\DecHam$ Diff}){$\kappa \in [k]$}{\label{line:dec-for}
			\textcolor{blue}{$\tildem\p{j_{\kappa}} \gets \BlockDecode^{\tC}(i_{\kappa})$} 
			
			\textcolor{blue}{\lIf{$\tildem\p{j_{\kappa}} = \bot$}{st $\pk_\kappa = \bot$ and \continue}}
			
			Parse $\tildem\p{j_{\kappa}}$ as $\tildex\p{j_\kappa} \circ \tsigma\p{j_\kappa} \circ \tpk\p{j_\kappa} \circ \tildej$
			
			Set $\pk_\kappa = \tpk\p{j_\kappa}$\label{line:pk-set}
		}
	
		Set $\pk^* = \mathsf{majority}(\pk_1,\dotsc, \pk_\mu)$\label{line:pk-majority}
	
		Compute $j \in [d]$ such that $(j-1) \cdot r(\lambda) < i \leq j \cdot r(\lambda)$
		
		\textcolor{blue}{$\tildem\p{j} \gets \NBS^{C'}(j)$}\tcp*{$\DecHam$ Diff}
		
		\textcolor{blue}{\lIf({\color{black}\tcp*[f]{$\DecHam$ Diff}}){$\tildem\p{j} = \bot$ or $\pk^* = \bot$}{\Return{$\bot$}}}
		
		Parse $\tildem\p{j}$ as $\tildex\p{j} \circ \tsigma\p{j} \circ \tpk\p{j} \circ \tildej$\label{line:tm-parse}
		
		\lIf{$\Verify_{\pk^*}(\tildex\p{j} \circ j, \tsigma\p{j}) = 0$}{\Return{$\bot$}}

		\Return{$\tildex\p{j}_{i^*}$ for $i^* = i-(j-1)\cdot r(\lambda)$}		
	\end{algorithm}
\paragraph*{\texorpdfstring{\cref{thm:insdel-crldc}}{Theorem 2} Proof Overview}
%We give a high-level overview of the proof here and refer the reader to the full version of our work \cite{BloBlo23} for the complete proof.
We give a high-level overview of the proof of \cref{thm:insdel-crldc}; full details can be found in \cref{sec:main-proof}.
The main technical challenge of proving \cref{thm:insdel-crldc} is showing that $\{(\EncIns,\DecIns)\}_\lambda$ satisfies \cref{item:rldc-fool,item:rldc-limit} of \cref{def:crldc}.
%We give a high level overview of the proof here.
%We present the formal proof of security in \cref{sec:main-proof} and give a high level overview in the remainder of this section.
Towards \cref{item:rldc-fool}, for any $x \in \bin{k}$, PPT adversary $\cA$, and $i \in [k]$, we analyze the probability that $\DecIns^{\tC}(i) \in \{x_i,\bot\}$ for $\tC \gets \cA(\EncIns(x))$ such that $\ED(\tC,\EncIns(x)) \leq \rho$.
The proof proceeds identically to the Hamming \crLDC with the following key changes.
First, when recovering the public key, we must consider the success probability of $\BlockDecode$ in our Chernoff bound to ensure  $\pk^*=\pk$ with high probability.
Second, we must consider the success probability of \NBS when recovering block $j$ that contains $x_i$. 
%Thus we additionally take into consideration the probability of success of our noisy binary search algorithm.
Careful selection of parameters and the guarantees of \BlockDecode and \NBS ensures \cref{item:rldc-fool} holds.
%As we show in \cref{sec:main-proof}, the guarantees of both \BlockDecode and \NBS allow us to ensure that \cref{item:rldc-fool} is satisfied via careful selection of parameters.

Towards \cref{item:rldc-limit}, the proof again proceeds nearly identically to the Hamming \crLDC case, except again we must take into consideration the recovery of public key $\pk^* = \pk$ via \BlockDecode and the recovery of block $j$ with \NBS.
The noisy binary search algorithm recovers any block that is $\gamma$-good with probability greater than $2/3$ (under suitable parameter choices), except with negligible probability.
This directly translates to the fraction $\delta = 1-\Theta(\rho)$ of indices we are able to decode from for \cref{item:rldc-limit}.
\section{Preliminaries}\label{sec:prelims}
We let $\lambda \in \bbN$ denote the security parameter.
For any $n \in \bbZ^+$ we let $[n] \defeq \{1,\dotsc, n\}$.
A function $\mu \colon \bbN \rightarrow \bbR_{\geq 0}$ is said to be negligible if $\mu(n) = o(1/|p(n)|)$ for any fixed non-zero polynomial $p$.
We write PPT to denote probabilistic polynomial time.
For any randomized algorithm $A$, we let $y \gets A(x)$ denote the process of obtaining output $y$ from algorithm $A$ on input $x$.
For a finite set $S$, we let $s \getsr S$ denote the process of sampling elements of $S$ uniformly at random.

We use ``$\circ$'' to denote the string concatenation operation.
For bitstring $x \in \bin{*}$, we use subscripts to denote individual bits of $x$; \eg, $x_i \in \zo$ is the $i$-th bit of $x$.
Additionally, we often partition a bitstring $x \in \bin{k}$ into some number of blocks $d$ of equal length; \eg, $x = (x\p{1} \circ \cdots \circ x\p{d})$ where $x\p{j} \in \bin{k/d}$ for all $j \in [d]$.
We also utilize array notation when convenient: \eg, for bitstring $x \in \bin{k}$ and indices $a, b \in [k]$ such that $a\leq b$, we let $x[a,b] \defeq (x_a \circ x_{a+1} \circ \cdots \circ x_{b})$.
For two strings $x \in \bin{k}$ and $y \in \bin{*}$, we define $\ED(x,y)$ as the minimum number of insertions and deletions required to transform $x$ into $y$ (or vice versa), normalized by $2k$.

In this work, we utilize \emph{digital signatures} and give the formal definition below.
\begin{definition}[Digital Signature Scheme]\label{def:sig}
%	Let $r(\cdot)$ be a polynomial.
%	An $(l, r, t, \epsilon)$-{\em (digital) signature scheme} is a tuple of PPT algorithms $\Pi = (\Gen, \Sign, \Verify)$ satisfying the following properties:
%	Let $r(\cdot)$, $t(\cdot)$, and $\epsilon(\cdot)$ be functions.
	A \emph{digital signature scheme} with signatures of length $r(\cdot)$ is a tuple of PPT algorithms $\Pi = (\Gen, \Sign, \Verify)$ satisfying the following properties:
	\begin{enumerate}
		\item $\Gen$ is the key generation algorithm and takes as input a security parameter $1^\lambda$ and outputs a pair of keys $(\pk,\sk) \in \bin{*} \times \bin{*}$, where $\pk$ is the public key and $\sk$ is the secret/private key.
		It is assumed that $|\pk|, |\sk| \geq \lambda$ are polynomial in $\lambda$, and that $\lambda$ can be efficiently determined from $\pk$ or $\sk$. 
%		For ease of presentation, we let $\pk(\lambda) \defeq |\pk|$ and $\sk(\lambda) \defeq |\sk|$ for $(\pk,\sk) \in \supp(\Gen(1^\lambda))$.
		Without loss of generality, we assume that $|\pk| = r(\lambda)$.
		
		\item $\Sign$ is the signing algorithm and takes as input secret key $\sk$ and message $m \in \bin{*}$ of arbitrary length and outputs a signature $\sigma \gets \Sign_\sk(m) \in \bin{r(\lambda)}$, where $\Sign$ runs in time $\poly(|\sk|, |m|)$.
		
		\item $\Verify$ is the deterministic verification algorithm that takes as input public key $\pk$, message $m$, and signature $\sigma$, and outputs a bit $b = \Verify_\pk(m,\sigma) \in \zo$.
		Moreover $\Verify$ run in time $\poly(r(\lambda), |m|)$.
	\end{enumerate}
	Additionally, we require the following two properties:
	\begin{enumerate}
		\item {\bfseries Completeness:} For all messages $m \in \bin{*}$ and all $(\pk, \sk) \in \supp(\Gen(1^\lambda))$, we have\footnote{Other definitions (\eg, \cite{KatzLin}) require this condition to hold except with negligible probability over $(\pk,\sk) \gets \Gen(1^\lambda)$.}
		\begin{align*}
			\Verify_\pk(m,\Sign_{\sk}(m)) = 1.
		\end{align*}
%		\alex{make a note of the difference between this definition and more standard ones.}
		
		\item {\bfseries Security:} For all PPT adversaries $\cA$, there exists a negligible function $\eps_{\Pi}(\cdot)$ such that for all $\lambda \in \bbN$ we have
		\begin{align*}
			\Pr[ \Sigforge_{\Pi, \cA}(\lambda) = 1 ] \leq \eps_{\Pi}(\lambda),
		\end{align*}
		where the experiment $\Sigforge$ is defined in \cref{fig:sig-forge}. 		
%		Define binary predicate $\Forge$ such that $\Forge(m, \sigma, \pk, Q) = 1$ iff (1) $\Verify_\pk(m,\sigma) = 1$; and (2) $m \not\in Q$.
%		We require that for all $\lambda \in \Zplus$ and all adversaries $\cA$ running in time $t(\lambda)$, the following holds:
%		\begin{align*}
%			\Pr[\Forge((\tilde{m}, \tilde{\sigma}), \pk, Q) = 1] \leq \epsilon(\lambda),
%		\end{align*} 
%		where the probability is taken over $(\pk, \sk) \gets \Gen(1^\lambda)$ and $(\tilde{m}, \tilde{\sigma}) \gets \cA^{\Sign_{\sk}(\cdot)}(\pk)$.
%		Here, the adversary is given oracle access to $\Sign_{\sk}(\cdot)$, and $Q$ is the set of all queries $\cA$ makes to its oracle before outputting $(\tilde{m}, \tilde{\sigma})$.
%		\alex{replace with game definition}
	\end{enumerate}
\end{definition}

% !TeX root = ../full-main.tex
\begin{figure}[tp]
\begin{boxedalgo}
	\begin{center}
		\ul{Signature experiment $\Sigforge_{\Pi, \cA}( \lambda )$}
	\end{center}
	\begin{enumerate}
		\item Obtain $(\pk, \sk) \gets \Gen(1^\lambda)$.
		
		\item Adversary $\cA$ is given $\pk$ as input and given oracle access to the algorithm $\Sign_{\sk}(\cdot)$.
		We denote this as $\cA^{\Sign_{\sk}(\cdot)}(\pk)$.
		Let $Q$ denote the set of all queries made by $\cA$ to its oracle.
		The adversary outputs pair $(m,\sigma)$.
		
		\item The adversary $\cA$ wins if and only if
		\begin{enumerate}
			\item $\Verify_{\pk^*}(m,\sigma) = 1$; and
			
			\item $m \not\in Q$.
		\end{enumerate}
		If $\cA$ wins, we define $\Sigforge_{\Pi, \cA}( \lambda ) \defeq 1$; otherwise we define $\Sigforge_{\Pi, \cA}( \lambda )\defeq 0$.
	\end{enumerate}
\end{boxedalgo}
\caption{Description of the signature forgery experiment $\Sigforge$.}\label{fig:sig-forge}
\end{figure}

For completeness, we also include the classical definition of an error-correcting code.
\begin{definition}
	A {\em coding scheme} $C[K,k,q_1,q_2] = (\Enc, \Dec)$ is a pair of encoding and decoding algorithms $\Enc\colon \Sigma_1^k \rightarrow \Sigma_2^K$ and $\Dec \colon \Sigma_2^* \rightarrow \Sigma_1^k$, where $|\Sigma_i| = q_i$.
	A code $C[K,k,q_1,q_2]$ is a {\em $(\rho,\dist)$ error-correcting code} for $\rho \in [0,1]$ and fractional distance $\dist$ if for all $x \in \Sigma_1^k$ and $y \in \Sigma_2^*$ such that $\dist(\Enc(x),y) \leq \rho$, we have that $\Dec(y) = x$.
	Here, $\rho$ is the {\em error rate} of $C$.
	If $q_1 = q_2$, we simply denote this by $C[K,k,q_1]$.
	If $\dist = \HAM$, then $C$ is a {\em Hamming code}; if $\dist = \ED$, then $C$ is an {\em insertion-deletion code} (InsDel code).
\end{definition}

Key to our construction is the so-called ``\SZ-code'', which is an insertion-deletion error-correcting code with constant rate and constant error-tolerance.
\begin{lemma}[SZ-code \cite{SchZuc99}]\label{lem:sz-code}
	There exists positive constants $\betasz \leq 1$ and $\rhosz > 0$ such that for large enough values of $t \in \Zplus$, there exists a $(\rhosz, \ED)$ code $\SZ(t) = (\SZ.\Enc,\SZ.\Dec)$ where $\SZ.\Enc \colon \bin{t}\rightarrow \bin{(1/\betasz) \cdot t}$ and $\SZ.\Dec \colon \bin{*} \rightarrow \bin{t} \union \{\bot\}$ with the following properties:
	\begin{enumerate}
		\item $\SZ.\Enc$ and $\SZ.\Dec$ run in time $\poly(t)$; and
		
		\item For all $x \in \bin{t}$, every interval of length $2\log(t)$ in $\Enc(x)$ has fractional Hamming weight $\geq 2/5$.
	\end{enumerate}
	We omit the parameter $t$ when it is clear from context.
\end{lemma}

% !TeX root = ../full-main.tex
\section{Proof of \texorpdfstring{\cref{thm:ham-crldc}}{Theorem 1.2}}\label{sec:ham-proof}
We dedicate this section to showing that $\cC_{\pred{H}, \lambda} = \{(\EncHam, \DecHam)\}_{\lambda \in \bbN}$ satisfies \cref{thm:ham-crldc}, where $\EncHam$ and $\DecHam$ are defined in \cref{alg:ham-encoder,alg:ham-decoder}, respectively.
We recall the theorem below.

\hamCRLDC*

\begin{proof}
First note that $x\p{j}, \sigma\p{j} \in \bin{r(\lambda)}$, and $j \in \bin{\log(d)}$ by construction.
Furthermore, without loss of generality we assume that $\pk \in \bin{\lambda}$ and $r(\lambda) \geq \lambda$.
Note we also have $\log(d) = O(\log(k))$.
Therefore if rate of $\Cin$ is $\betain$, then block length $K$ of our code is exactly
\begin{align*}
	K&=d \cdot (1/\betain) \cdot ( 2r(\lambda) + \lambda + \log(d) )\\ &=(\ceil{k/r(\lambda)}) \cdot (1/\betain) \cdot (2r(\lambda) + \lambda + \log(\ceil{k/r(\lambda)}))\\
	&=O((1/\betain) \cdot k \cdot (3 + \log(k)/r(\lambda) ))\\
	&=O((1/\betain)\cdot k(1 + \log(k)/r(\lambda)))
\end{align*}
%$d \cdot (1/\betain) \cdot ( 2r(\lambda) + \lambda + \log(d) ) = (\ceil{k/r(\lambda)}) \cdot (1/\betain) \cdot (2r(\lambda) + \lambda + \log(\ceil{k/r(\lambda)})) = O((1/\betain) \cdot k \cdot (3 + \log(k)/r(\lambda) )) = O((1/\betain)\cdot k(1 + \log(k)/r(\lambda)))$ 
whenever $k \geq r(\lambda)$.
When $r(\lambda) > k$, then we pad the input message $x$ with $k-r(\lambda)$ number of $0$'s at the end to get a string of length $r(\lambda)$, which gives a single codeword block of length
\begin{align*}
	K &= (1/\betain) \cdot (2r(\lambda) + \lambda + 1)\\
	&= O((1/\betain)\cdot r(\lambda)).
\end{align*}
Thus we have our specified block length 
\begin{align*}
	K = O((1/\betain)\cdot \max\{k\cdot (1 + \log(k)/r(\lambda)), r(\lambda)\}).
\end{align*}
For the locality $\ell$, by construction we know that any block $j$ has length
\begin{align*}
	&(1/\betain) \cdot (2r(\lambda) + \lambda + \log(d))\\
	&=(1/\betain) \cdot (2r(\lambda) + \lambda + \log(k)-\log(r(\lambda)))\\
	&=O((1/\betain)\cdot (r(\lambda) + \log(k))).
\end{align*} %$(1/\betain) \cdot (2r(\lambda) + \lambda + \log(d)) = (1/\betain) \cdot (2r(\lambda) + \lambda + \log(k)-\log(r(\lambda))) = O((1/\betain)\cdot (r(\lambda) + \log(k)))$.
Since we decode $\mu+1$ blocks, our overall locality is 
\begin{align*}
	\ell = O((\mu/\betain)\cdot (r(\lambda) + \log(k))).
\end{align*}
Moreover, $\DecHam$ makes $O((\mu/\betain)\cdot (r(\lambda) + \log(k)))$ queries to its oracle on any input $i$, satisfying \cref{item:rldc-query} of \cref{def:crldc}.
%Note that $\mu \in \bbN$ is a parameter of our choosing used to ensure recovery of $\pk$ correctly with high probability via majority sampling.

For item \cref{item:rldc-correct}, assume that $\DecHam$ is given oracle access to $\tC = C$ for some $C = \EncHam(x)$ and $x \in \bin{k}$.
We then analyze the probability that $\DecHam^{C}(i) = x_i$ for any $i \in [k]$.
First, since $C$ is a correct codeword, recovery of the public key succeeds with probability $1$.
That is, for every $\kappa \in [\mu]$, the string $\tildem\p{j_{\kappa}} \gets \Decin\left(C[(j_\kappa-1)\cdot \bl+1, j_\kappa \cdot \bl]\right)$ recovered in \cref{line:decham-pk-decode} is equal to $x\p{j_\kappa} \circ \sigma\p{j_\kappa} \circ \pk \circ j_\kappa$ (\ie, everything is correct).
Here, $\bl = n/d$ is the length of each block $C\p{j}$ in $C = C\p{1} \circ \cdots \circ C\p{d}$.
Thus $\pk^* = \pk$ with probability $1$.
Now fixing $j \in [d]$ to be the block such that bit $x_i$ resides in $x\p{j}$, by the above discussion we know that $\tildem\p{j}$ recovered in \cref{line:decham-decode} is correct and is parsed as $x\p{j} \circ \sigma\p{j} \circ \pk \circ j$.
This along with the fact that $\pk^* = \pk$ implies that \cref{line:decham-verify} is true with probability $0$ (\ie, $\Verify_{\pk^*}(x\p{j}\circ j, \sigma\p{j}) = 1$ with probability $1$).
This implies that $\DecHam^{C}(i) = x_i$ with probability $1$. %, as desired.

For error-rate, let $\rhoin \in (0,1)$ be the error-rate of $\Cin$. 
Intuitively, we want to set our final error-rate $\rho$ such that for any $\rho$-fraction of corruptions, less than half of the $d$ blocks contain more than $\rhoin$ faction of errors each.
Equivalently, more than half of the $d$ blocks contain at most $\rhoin$-fraction of errors.
Let $\bl$ denote the length of a block in the codeword $C$.
Then $\Decin$ cannot correctly decode any block $j \in [d]$ if it has at least $\rhoin \bl + 1$ errors.
Let $\tJ \subset [d]$ be the set of indices such that block $j \in \tJ$ has at least $\rhoin \bl + 1$ Hamming errors.
Then we have 
\begin{align*}
	|\tJ|\cdot (\rhoin \bl +1) \leq \rho\cdot d \cdot \bl ~(= \rho K),	
\end{align*}
which implies that 
\begin{align*}
	|\tJ| &\leq (\rho\cdot d\cdot \bl)/(\rhoin \bl + 1)\\
	&< (\rho\cdot d\cdot \bl)/(\rhoin \bl)\\ 
	&= \rho d / \rhoin.
\end{align*} 
We want $|\tJ| < d/2$ to ensure that more than half of the blocks contain at most $\rhoin \bl$ errors.
Setting $(d\rho)/\rhoin < d/2$ implies $\rho< \rhoin/2$.
Thus we set $\rho = c \cdot \rhoin = \Theta(\rhoin)$ for any constant $c  \in (0,1/2)$.

Next we analyze the success probability $p$ and \cref{item:rldc-fool}.
For predicate $\Fool$, fix message $x$ and let $C = \Enc_{\pred{H},\lambda}(x)$.
We want to show that for all PPT $\cA$ there exists a negligible function $\eps_{\F}$ such that for all $\lambda$ and $x \in \bin{k}$ we have 
\begin{align*}
	\Pr[\Fool(\tC, \rho, p, x, C, \lambda) = 1] \leq \eps_{\F}(\lambda)
\end{align*} 
for $\tC \gets \cA(C)$.
Equivalently, we want to show that %$\Pr[ \exists i \in [k] \colon \Pr[\Dec_{\pred{H},\lambda}^{\tC}(i) \in \{x_i,\bot\}] < p ] \leq \eps_{\F}(\lambda)$
\begin{align*}
	\Pr[ \exists i \in [k] \colon \Pr[\Dec_{\pred{H},\lambda}^{\tC}(i) \in \{x_i,\bot\}] < p ] \leq \eps_{\F}(\lambda)
\end{align*}
for $\tC\gets\cA(C)$ such that $\HAM(C,\tC) \leq \rho$.
Restated again, we want to show %$\Pr[ \forall i \in [k] \colon \Pr[\Dec_{\pred{H},\lambda}^{\tC}(i) \in \{x_i,\bot\}] \geq p ] \geq 1- \eps_{\F}(\lambda)$.
\begin{align*}
	\Pr[ \forall i \in [k] \colon \Pr[\Dec_{\pred{H},\lambda}^{\tC}(i) \in \{x_i,\bot\}] \geq p ] \geq 1- \eps_{\F}(\lambda).
\end{align*}

We analyze the probability that $\Dec_{\pred{H},\lambda}^{\tC}(i) \in \{x_i,\bot\}$, and let $E_i$ denote this event, and also let $\Forge_i$ denote the event that $\cA$ produces a signature forgery $(\tildex\p{j} \circ j, \tsigma\p{j})$, where $j \in [d]$ such that $(j-1)\cdot r(\lambda) < i \leq j\cdot r(\lambda)$ and $\tildex\p{j}$ and $\tsigma\p{j}$ are recovered by $\DecHam$. % \cref{line:decham-decode} of \cref{alg:ham-decoder}.
Then we have that $\Pr[E_i] = \Pr[E_i | \overline{\Forge_i}]$.
%\begin{align*}
%	\Pr[E_i] &= \Pr[E_i ~|~ \overline{\Forge_i}].
%\end{align*}
Note that the decoder can never output $\bot$ and $x_i$ simultaneously.
Letting $E_i(x)$ be the event that $\DecHam^{\tC}(i) = x$ for symbol $x$, we have that %$\Pr[E_i | \overline{\Forge_i}] = \Pr[E_i(x_i) | \overline{\Forge_i}] + \Pr[E_i(\bot) | \overline{\Forge_i}]$.
\begin{align*}
	\Pr[E_i ~|~ \overline{\Forge_i}] = \Pr[E_i(x_i) ~|~ \overline{\Forge_i}] + \Pr[E_i(\bot) ~|~ \overline{\Forge_i}].
\end{align*}
Analyzing $\Pr[E_i(x_i) | \overline{\Forge_i}]$, since we assume that $(\tildex\p{j}\circ j, \tsigma\p{j})$ is not a forgery it must be the case that
\begin{enumerate}
	\item $\tildex\p{j} = x\p{j}$ and $\tsigma\p{j} = \sigma\p{j}$; and
	\item $\Verify_{\pk^*}(\tildex\p{j} \circ j, \tsigma\p{j}) = 1$.
\end{enumerate}
Now this verification only succeeds if $\pk^* = \pk$, which implies %$\Pr[E_i(x_i) | \overline{\Forge_i}] = \Pr[\pk^* = \pk]$.
\begin{align*}
	\Pr[E_i(x_i) ~|~ \overline{\Forge_i}] = \Pr[\pk^* = \pk].
\end{align*}

Next we analyze $\Pr[E_i(\bot) | \overline{\Forge_i}]$.
Note that $\DecHam$ only outputs $\bot$ if $\Verify_{\pk^*}(\tildex\p{j}\circ j, \tsigma\p{j}) = 0$.
Let $\cE_1$ denote the event $\overline{\Forge_i}$, $\cE_2$ denote the event $\tildex\p{j} \neq x\p{j}$, $\cE_3$ denote the event $\tsigma\p{j} \neq \sigma\p{j}$, and $\cE_4$ denote the event $\pk^* = \pk$.
Then we lower bound the above probability as %$\Pr[E_i(\bot) | \overline{\Forge_i}] \geq \Pr[ \Verify_{\pk^*}(\tildex\p{j}\circ j, \tsigma\p{j}) = 0 |  \overline{\Forge_i} \band (\tildex\p{j} \neq x\p{j} \bor \tsigma\p{j} \neq \sigma\p{j}) \band \pk^* = \pk] \cdot \Pr[\pk^* = \pk]$.
\begin{align*}
	&\Pr[E_i(\bot) | \overline{\Forge_i}] \geq\\
	&\Pr\left[ \Verify_{\pk^*}(\tildex\p{j}\circ j, \tsigma\p{j}) = 0 \bigg| \begin{array}{l}
		\cE_1 \band \\
		(\cE_2 \bor \cE_3) \band \\
		\cE_4
	\end{array}\right] \cdot \Pr[\cE_4].
\end{align*}
Since we assume $\overline{\Forge_i}$ is true (\ie, $\cE_1$ is true), we know that 
\begin{align*}
	\Pr\left[ \Verify_{\pk^*}(\tildex\p{j}\circ j, \tsigma\p{j}) = 0 \bigg| \begin{array}{l}
		\cE_1 \band \\
		(\cE_2 \bor \cE_3) \band \\
		\cE_4
	\end{array}\right] = 1.
\end{align*}
%$\Pr[ \Verify_{\pk^*}(\tildex\p{j}\circ j, \tsigma\p{j}) = 0 |  \overline{\Forge_i} \band (\tildex\p{j} \neq x\p{j} \bor \tsigma\p{j} \neq \sigma\p{j}) \band \pk^* = \pk] = 1$.
This implies %$\Pr[E_i(\bot) | \overline{\Forge_i}] \geq \Pr[\pk^* = \pk]$,
\begin{align*}
	\Pr[E_i(\bot) ~|~ \cE_1] \geq \Pr[\cE_4] = \Pr[\pk^* = \pk],
\end{align*}
which implies %$\Pr[E_i | \overline{\Forge_i}] \geq 2 \cdot \Pr[\pk^* = \pk] \geq \Pr[\pk^* = \pk]$.
\begin{align*}
	\Pr[E_i ~|~ \overline{\Forge_i}] \geq 2 \cdot \Pr[\pk^* = \pk] \geq \Pr[\pk^* = \pk].
\end{align*}

We now analyze $\Pr[\pk^* = \pk]$.
Let $\pk$ be the public key sampled by $\Enc_{\pred{H},\lambda}(x)$ to generate $C$.
By our parameter choice $c \in (0,1/2)$, at least $(1-c) d > d/2$ blocks of $C'$ contain at most $\rhoin$-fraction of Hamming errors.
Let $\bl = K/d$ denote the length of any block of $C$ and let $\cJ$ denote the set of blocks with at most $\rhoin$-fraction of Hamming errors.
Then 
\begin{align*}
	x\p{j} \circ \sigma\p{j}\circ \pk \circ j = \Decin(\tC[(j-1)\cdot\bl+1, j\cdot\bl])
\end{align*} 
for any $j \in \cJ$; \ie, it is a correct decoding since corrupt codeword $\tC\p{j} = \tC[(j-1)\cdot \bl+1, j\cdot \bl]$ is within the unique decoding radius of $\Cin$.
Define random variable $X_\kappa$ for $\kappa \in [\mu]$ as %$X_\kappa = 1$ if $x\p{j}\circ\sigma\p{j}\circ\pk\circ j = \Decin( \tC[(j-1) \bl+1, j \bl] )$ and $X_\kappa = 0$ otherwise.
\begin{align*}
	X_\kappa \defeq \begin{cases}
			1 & x\p{j}\circ\sigma\p{j}\circ\pk\circ j = \Decin( \tC[(j-1) \bl+1, j \bl] )  \\
			0 & \text{otherwise}
		\end{cases}.
\end{align*}
%where $C\p{j_\kappa}$ is the $j_\kappa$-th block of uncorrupt codeword $C$.
Then %$\Pr[X_\kappa = 1] = \Pr[x\p{j}\circ\sigma\p{j}\circ\pk\circ j = \Decin( \tC[(j-1)\cdot \bl+1, j \cdot \bl] = \Pr[ j_\kappa \in \cJ ] )] \geq (1-c) > 1/2$.
\begin{align*}
	&\Pr[X_\kappa = 1] \\
	&= \Pr\left[x\p{j}\circ\sigma\p{j}\circ\pk\circ j = \Decin( \tC[(j-1)\cdot \bl+1, j \cdot \bl] )\right]\\
	&= \Pr[ j_\kappa \in \cJ_{\Good} ]\\
	&\geq (1-c) > 1/2.
\end{align*}
Let $q = (1-c) > 1/2$.
%Choose parameter $\epsilon \in (0,1/2)$ and set $\mu \defeq \log(1/\epsilon) \cdot (2q)/(q-1/2)^2$.
By Chernoff bound, we have that %$\Pr[ \sum_{\kappa \in [\mu]} X_\kappa > \mu/2 ] \geq 1-\exp(-\mu\cdot(q-1/2)^2/(2q))$.
\begin{align*}
	\Pr\left[ \sum_{\kappa \in [\mu]} X_\kappa > \frac{\mu}{2} \right] \geq 1-\exp(-\mu\cdot(q-1/2)^2/(2q)).
\end{align*}
This implies that with probability at least 
\begin{align*}
	p \defeq 1-\exp(-\mu\cdot(q-1/2)^2/(2q)),
\end{align*} 
we have $\pk^* = \pk$.
Thus %$\Pr[E_i | \overline{\Forge_i}] \geq p$.
\begin{align*}
	\Pr[E_i ~|~ \overline{\Forge_i}] \geq p.
\end{align*}

Throughout the above analysis, we only assumed that no forgery occurred.
For any arbitrary $i$, the probability that $\overline{\Forge_i}$ occurs is at least $1-\eps_{\Pi}(\lambda)$, where $\eps_{\Pi}(\cdot)$ is a negligible function that depends on the security of the digital signature scheme $\Pi$.
Thus by union bound over all $i$, we have %$\Pr[ \exists i \in [k] \colon \Pr[\Dec_{\pred{H},\lambda}^{\tC}(i) \in \{x_i,\bot\}] < p ] \geq k \cdot \eps_{\Pi}(\lambda)$. 
\begin{align*}
	\Pr[ \exists i \in [k] \colon \Pr[\Dec_{\pred{H},\lambda}^{\tC}(i) \in \{x_i,\bot\}] < p ] \geq k \cdot \eps_{\Pi}(\lambda).
\end{align*}
Setting $\eps_{\F}(\lambda) \defeq k \cdot \eps_{\Pi}(\lambda)$, we have that $\eps_{\F}(\lambda)$ is negligible in $\lambda$ since $k(\lambda)$ is a polynomial, showing \cref{item:rldc-fool}.
%\alex{this is no longer correct. If $k$ is a parameter of the family, this can be non-negligible. If $k$ is a polynomial in $\lambda$, then this holds. Maybe we should just fix $k$ to be a polynomial in $\lambda$ to make things simple.}
%\alex{Do we need $k$ to be a polynomial in $\lambda$ here? I think we don't since $k \in \bbN$ is just some sufficiently large, fixed integer.}

Finally, we analyze $\delta$ and \cref{item:rldc-limit}. % of \cref{def:crldc}.
By our choice of $\rho = c \cdot \rhoin$, we know that for \emph{any} $\tC \in \bin{n}$ such that $\HAM(\tC, C) \leq \rho$, at least $(1-c) \cdot d$ blocks of $\tC$ contain at most $\rhoin$-fraction of Hamming errors.
Again letting $\cJ \subset [d]$ denote the indices of these blocks, we have $|\cJ| > d/2$.
For $\bl = K/d$, this again implies that 
\begin{align*}
	x\p{j} \circ \sigma\p{j}\circ \pk \circ j = \Decin(\tC[(j-1)\cdot \bl+1, j\cdot \bl])
\end{align*} 
for any $j \in \cJ$, which implies for any $j \in \cJ$ we have that 
\begin{align*}
	\Pr[ \DecHam^{\tC}(i) = x_i | \pk^* = \pk ] = 1
\end{align*} 
whenever $(j-1) \cdot r(\lambda) < i \leq j\cdot r(\lambda)$.
Thus for any $j \in \cJ$ and $i \in [k]$ such that $(j-1) \cdot r(\lambda) < i \leq j \cdot r(\lambda)$, by Chernoff we have that %$\Pr[ \DecHam^{\tC}(i) = x_i ] = \Pr[ \pk^* = \pk ] \geq 1- \exp(-\mu \cdot (1/2-c)^2 \big/ 2(1-c))$ by Chernoff.
\begin{align*}
	\Pr[ \DecHam^{\tC}(i) = x_i ] &= \Pr[ \pk^* = \pk ]\\
	&\geq 1- \exp(-\mu \cdot (1/2-c)^2 \big/ 2(1-c)).
\end{align*}
By appropriate choice of $\mu$, we can ensure that 
\begin{align*}
	1- \exp(-\mu \cdot (1/2-c)^2 \big/ 2(1-c)) > 2/3.
\end{align*}
Finally, we note that the set
\begin{align*}
	\cI \defeq \{ i\in [k] \colon j \in \cJ \band (j-1) \cdot r(\lambda) < i \leq j\cdot r(\lambda)\}
\end{align*} 
has size $|\cI| > k/2$.
Thus we can set $\delta = 1/2$.

Here we have shown that for \emph{any} $\tC \in \bin{n}$ such that $\HAM(\tC, C) \leq \rho$, there exists a set $\Good(\tC) \defeq \cI$ such that $|\cI| \geq \delta \cdot k$.
This is for any $\tC$, and in particular, any corrupt codeword that a PPT adversary could produce.
Thus we have that for any PPT adversary $\cA$, any $x \in \bin{k}$, and $C = \EncHam(x)$, we have 
\begin{align*}
	\Pr[ \Limit( \cA(C), \rho, \delta, x, C, \lambda) = 1 ] = 0.
\end{align*}
\end{proof}
%Let $j \in [d]$ be such that $(j-1) \cdot r(\lambda) < i \leq j\cdot r(\lambda)$.
%Then block $\tC\p{j} = \tC[ (j-1)r(\lambda)+1, jr(\lambda) ]$ falls into two error cases.
%In the first case, $j \in \cJ_{\Good}$.
%In this case, by assumption $\Decin(\tC\p{j})$ uniquely recovers $x\p{j}\circ \sigma\p{j}\circ \pk \circ j$.
%So long as $\pk^* = \pk$, this implies that verification succeeds with probability $1$, and thus in this case we output the correct bit $x_i$ with probability at least $1-\epsilon$.
%Now suppose that $j \not\in \cJ_{\Good}$.
%This implies that $\Decin(\tC\p{j})$ outputs some corrupt $\tildex \circ \tsigma \circ \tpk \circ \tildej \neq x\p{j} \circ \sigma\p{j} \circ \pk \circ j$.
%Note that if $\tildex = x\p{j}$ and $\tsigma = \sigma\p{j}$, then with probability at least $1-\epsilon$ we output the correct bit $x_i$.
%Thus suppose that $\tildex \neq x\p{j}$ or $\tsigma \neq \sigma\p{j}$.
%In this case, if the verification succeeds and $\pk^* = \pk$, then we have a strong forgery.
%Security of our signature scheme says this happens with probability at most $\negl(\lambda)$.
%So with probability at least $1-\negl(\lambda)$, if $\pk^* = \pk$ then the verification will fail and the decoder will output $\bot$.

% !TeX root = ../full-main.tex
\section{Proof of \texorpdfstring{\cref{thm:insdel-crldc}}{Theorem 2}}\label{sec:main-proof}
We dedicate this section to showing that $\cC_{\pred{I},\lambda} = \{(\EncIns,\DecIns)\}_{\lambda \in \bbN}$ satisfies \cref{thm:insdel-crldc}, where $\EncIns$ and $\DecIns$ are defined in, \cref{alg:encoder,alg:decoder}, respectively.
We recall the theorem below.
\mainthm*
\begin{proof}
	Fix any $x \in \bin{k}$ and let $C = \EncIns(x)$.
	In the definition of $\EncIns$, we know that $C = C\p{1} \circ \cdots \circ C\p{d}$ for $d = \ceil{k / r(\lambda)}$.
	First assume that $k \geq r(\lambda)$.
	Each block $C\p{j}$ is the $\SZ$ encoding of $x\p{j} \circ \sigma\p{j} \circ \pk \circ j$, appended at the front and back with zero-buffers.
	Note that the bit-length of $j$ is $\log(d) = \log(k) - \log(r(\lambda))$; 
	for simplicity, we assume $j \in \bin{\log(k)}$ (we can pad to length $\log(k)$ otherwise).
	%Then we have that $|x\p{j}\circ j| = 2\log(k) \leq l(\lambda)$ since we compute signature $\sigma\p{j} \gets \Sign_{\sk}(x\p{j} \circ j)$; thus we simply define $k = k(\lambda)$ such that $2\log(k) = l(\lambda)$.
	Then we have $\tau \defeq |x\p{j} \circ \sigma\p{j} \circ \pk \circ j| = 3 \cdot r(\lambda) + \log(k)$. 
	This gives us that $|c\p{j}| = \tau/\betasz$, and that
	\begin{align*}
		|C\p{j}| &= 2\alpha \tau + \tau/\betasz\\ 
		&= (2\alpha + (1/\betasz))\tau.
	\end{align*}
	Finally, this gives %$|C| = K = d \cdot (2\alpha+(1/\betasz))\cdot \tau = (k/r(\lambda)) \cdot (2\alpha + (1/\betasz)) \cdot ( 3 \cdot r(\lambda) + \log(k) ) = O(k\cdot(1+\log(k)/r(\lambda)))$,
	\begin{align*}
		|C| = K &= d \cdot (2\alpha+(1/\betasz))\cdot \tau\\
		&= \left(\frac{k}{r(\lambda)}\right) \cdot (2\alpha + (1/\betasz)) \cdot ( 3 \cdot r(\lambda) + \log(k) )\\
		&= O\left( k \cdot \left(1 + \frac{\log(k)}{r(\lambda)}\right) \right),
	\end{align*}
	where the last equality holds since $\alpha, \betasz$ are constants.
	Note here that $k$ is sufficiently large whenever $\SZ(t)$ exists for all $t \geq \log(k)$.
	Now whenever $r(\lambda) > k$, we have that $d=1$ and we simply sign and encode a single block of length $\tau$, which yields a single block of length 
	\begin{align*}
		(2\alpha + (1/\betasz))\tau &= (2\alpha + (1/\betasz))\cdot(3r(\lambda) + \log(k))\\ 
		&= \Theta(r(\lambda))
	\end{align*} 
	since $r(\lambda) > k$ and $\alpha, \betasz$ are constants.
	Thus we have 
	\begin{align*}
		K = O( \max\{ k(1+\log(k)/r(\lambda)), r(\lambda) \} ).
	\end{align*}
	
	For the remainder of the proof, let $\beta \defeq (2\alpha + (1/\betasz))$ and $\tC \in \bin{K'}$ such that $\ED(C, \tC) \leq \rho$. 
	We introduce some preliminary definitions (\crefrange{def:block-decomp}{def:gamma-good}) and lemmas \crefrange{lem:block-decomp-bound}{lem:block-decode}), due to Block \etal \cite{BBGKZ20}, needed for the proof. %, before continuing with the proof of \cref{thm:insdel-crldc}.
%	These definitions (\crefrange{def:block-decomp}{def:gamma-good}) and lemmas (\crefrange{lem:block-decomp-bound}{lem:block-decode}) come from the results of Block \etal \cite{BBGKZ20}.
%	First we define the notion of a \emph{block decomposition}.

	\begin{definition}\label{def:block-decomp}
		A \emph{block decomposition} of a (corrupt) codeword $\tC \in \bin{K'}$ is a non-decreasing map $\phi \colon [K'] \rightarrow [d]$ for $K', d \in \bbN$.
	\end{definition}
	\noindent For any block decomposition $\phi$, since $\phi$ is a non-decreasing map we have that $\phi^{-1}(j)$ for any $j \in [d]$ is an interval.
	That is, $\phi^{-1}(j) = \{l_j, l_j+1, \dotsc, r_j\}$ for integers $l_j, r_j \in [K']$ and $l_j \leq r_j$.
	Thus $\phi$ induces a partition of $[K']$ into $d$ intervals of the form $\{ \phi^{-1}(j) \colon j \in [d] \}$.
	%Another necessary notion is that of a closure interval.
	%\begin{definition}[\cite{BBGKZ20}]
	%	The \emph{closure} of an interval $I = \{l, l+1,\dotsc, r-1\} \subseteq [n']$ is defined as $\cup_{i \in I} \phi^{-1}( \phi(i) )$.
	%	In interval $I$ is a \emph{closure interval} if the closure of $I$ is itself.
	%	Equivalently, every closure interval has the form $\cI[a,b] \defeq \bigcup_{j=a}^b \phi^{-1}(j)$ for some $a,b, \in [d]$.
	%\end{definition}
	Recalling that $C = \EncIns(x)$ is of the form $C\p{1}\circ \cdots \circ C\p{d}$, we have the following.
	%We have the following lemma from \cite{BBGKZ20}.
	\begin{lemma}\label{lem:block-decomp-bound}
		There exists a block decomposition $\phi_0 \colon [K'] \rightarrow [d]$ such that %$\sum_{j \in [d]} \ED( \tC[\phi_0^{-1}(j)], C\p{j} ) \leq \rho$.
		\begin{align}
			\sum_{j \in [d]} \ED\left( \tC[\phi_0^{-1}(j)], C\p{j} \right) \leq \rho\; .\label{eq:block-decomp-bound}
		\end{align}
	\end{lemma}
	\noindent Intuitively, \cref{lem:block-decomp-bound} says there exists a block decomposition such that the total edit distance between $C'$ and $C'$ is exactly given by the sum of edit distances between the (possibly corrupt) blocks $\tC[\phi_0^{-1}(j)]$ and blocks $C\p{j}$.
	Next we define the notion of a $\gamma$-good block.
	
	\begin{definition}\label{def:gamma-good}
		For $\gamma \in (0,1)$ and $j \in [d]$, we say that block $j$ is \emph{$\gamma$-good} with respect to a block decomposition $\phi$ if $\ED(\tC[\phi^{-1}(j)], C\p{j}) \leq \gamma$.
		Otherwise, we say block $j$ is \emph{$\gamma$-bad}.
	\end{definition}
	%\noindent We also need a notion of a ``locally good'' block.
	%The following two definitions give us this notion.
	%\begin{definition}[$(\theta,\gamma)$-good interval \cite{BBGKZ20}]\label{def:local-interval}
	%	A closure interval $\cI[a,b]$ is \emph{$(\theta,\gamma)$-good} if the following hold.
	%	\begin{enumerate}
		%		\item $\sum_{j=a}^b \ED( \tC[\phi^{-1}_0(j)], C\p{j} ) \leq \gamma \cdot (b-a+1)\cdot \alpha \cdot \tau$.
		%		
		%		\item There are at least $(1-\theta)$-fraction of $\gamma$-good blocks in the set $\{ \tC[\phi^{-1}_0(j)] \colon j \in \cI[a,b] \}$.
		%	\end{enumerate}
	%\end{definition}
	
	%\begin{definition}[$(\theta,\gamma)$-locally good block]\label{def:local-good}
	%	For $\theta, \gamma \in (0,1)$, we say that a block $j \in [d]$ is \emph{$(\theta,\gamma)$-locally good} if for every $a,b \in [d]$ such that $a \leq j \leq b$, the interval $\cI[a,b]$ is $(\theta,\gamma)$-locally good.
	%	Otherwise, block $j$ is $(\theta,\gamma)$-locally bad.
	%\end{definition}
	%
	%
	
	\noindent With respect to block decomposition $\phi_0$, the number of $\gamma$-bad blocks is bounded, and the length of the intervals $\phi_0^{-1}(j)$ is bounded for every $\gamma$-good block $j$.
	%\noindent It was shown that with respect to the block decomposition $\phi_0$, the total number of $\gamma$-bad blocks is bounded.
	
	\begin{lemma}\label{lem:gamma-block-bounds}
		Let $\alpha$ be the constant given by \cref{lem:nbs}, let $\betasz$ be the constant given by \cref{lem:sz-code}, and let $\beta = (2\alpha + (1/\betasz))$.
%		Then the block decomposition $\phi_0$ satisfies the following properties. (1) The total fraction of $\gamma$-bad blocks in $\tC$ is at most $2 \beta \rho / (\gamma \alpha)$; (2) For any $\gamma$-good block $j$, we have that $(\beta - \alpha\gamma) \cdot \tau \leq |\phi_0^{-1}(j)| \leq (\beta + \alpha\gamma) \cdot \tau$, where $\tau = \abs{x\p{j} \circ \sigma\p{j} \circ \pk \circ j}$.
		\begin{enumerate}
			\item The total fraction of $\gamma$-bad blocks in $\tC$ is at most $2\cdot \beta \cdot \rho / (\gamma \cdot \alpha)$.
			
			%		\item The total fraction of $(\theta,\gamma)$-bad blocks in $\tC$ is at most $(4 \cdot \beta \cdot \rho) \cdot (1 + 1/\theta) /(\gamma \cdot \alpha)$.
			
			\item For any $\gamma$-good block $j$, we have that $(\beta - \alpha\gamma) \cdot \tau \leq |\phi_0^{-1}(j)| \leq (\beta + \alpha\gamma) \cdot \tau$, where $\tau = |x\p{j} \circ \sigma\p{j} \circ \pk \circ j|$.
			
			%		\item For any $(\theta,\gamma)$-good interval $\cI[a,b]$, we have $(b-a+1)\cdot (\beta-\alpha\gamma)\tau \leq |\cI[a,b]| \leq (b-a+1)\cdot (\beta+\alpha\gamma)\cdot \tau$.
		\end{enumerate}
	\end{lemma}
	
	\noindent Given the notion of $\gamma$-good, we can now formally introduce the algorithms \NBS and \BlockDecode along with their guarantees.
	
	\begin{lemma}\label{lem:nbs}
		Let $\rhosz$ be the constant given by \cref{lem:sz-code}.
		There exists constant $\alpha = \Omega(\rhosz)$ and a randomized oracle algorithm $\NBS$ with the following property.
		Let $\rho^* \in (0,1/2)$ be a fixed constant, let $t$ be sufficiently large, $d$ be a parameter, let $b = (b\p{1},\dotsc, b\p{d}) \in \bin{t}$ be any string where $b\p{i} \in \bin{t/d}$ for all $i \in [d]$, and let $c = (c\p{1},\dotsc, c\p{d})$ for 
		\begin{align*}
			c\p{i} = 0^{\alpha (t/d)} \circ \SZ.\Enc(b\p{i} \circ i) \circ 0^{\alpha (t/d)}
		\end{align*} 
		for all $i \in [d]$.
		Then there exists a negligible function $\vartheta(\cdot)$ such that for any $c' \in \bin{K'}$ satisfying $\ED(c,c') \leq \rho = \Theta(\rho^* \cdot \rhosz)$, we have that %$\Pr[ \Pr[\NBS^{c'}(j) \neq b\p{j}| j \text{ is $\gamma$-good}] \geq \rho^* ] \leq \vartheta(K')$, 
		\begin{align}
			\Pr[ \Pr[\NBS^{c'}(j) \neq b\p{j}~|~ j \text{ is $\gamma$-good}] \geq \rho^* ] \leq \vartheta(n'),\label{eq:nbs}
		\end{align}
		where the probability is taken over the random coins of $\NBS$ and $j \getsr [d]$.
		Furthermore, the algorithm $\NBS$ makes $O(\log^3(K') \cdot (t/d+\log(d)))$ oracle queries for any input $j \in [d]$, and if $c = c'$ then the above probability is $0$.
	\end{lemma}
	
	\begin{lemma}\label{lem:block-decode}
		Let $\rhosz$ be the constant given by  \cref{lem:sz-code}.
		There exists constant $\alpha = \Omega(\rhosz)$ and randomized oracle algorithm $\BlockDecode$ with the following properties.
		Let $\rho^* \in (0,1/2)$ be a fixed constant, let $t$ be sufficiently large, let $d$ be a parameter, let $b = (b\p{1}, \dotsc, b\p{d})\in \bin{t}$ be any string where $b\p{i} \in \bin{t/d}$ for all $i\in[d]$, and let $c = (c\p{1},\dotsc, c\p{d})$ for 
		\begin{align*}
			c\p{i} = 0^{\alpha (t/d)} \circ \SZ.\Enc(b\p{i} \circ i) \circ 0^{\alpha (t/d)}
		\end{align*} 
		for all $i \in [d]$.
		Then for any $c' \in \bin{n'}$ satisfying $\ED(c,c') \leq \rho = \Theta(\rho^* \cdot \rhosz)$, we have that %(1) For any $\gamma$-good block $j \in [d]$, $\Pr[ \BlockDecode^{c'}(i) \neq b\p{j} ] \leq \gamma$ for $i \in \phi_0^{-1}(j)$, where this probability equals zero if $c' = c$; (2) \BlockDecode has query complexity $O(t/d+\log(d))$.
		\begin{enumerate}
			\item For any $\gamma$-good block $j \in [d]$: %, $\Pr[ \BlockDecode^{c'}(i) \neq b\p{j} ] \leq \gamma$ for $i \in \phi_0^{-1}(j)$.
			\begin{align}
				\Pr_{i \in \phi_0^{-1}(j)}\left[ \BlockDecode^{c'}(i) \neq b\p{j} \right] \leq \gamma\; . \label{eq:block-decode}
			\end{align}
			Furthermore, this probability is equal to zero if $c' = c$.
			
			\item \BlockDecode has query complexity $O(t/d+\log(d))$.
		\end{enumerate}
	\end{lemma}
	
	We now have the necessary components to finish proving \cref{thm:insdel-crldc}.
	We begin by showing \cref{item:rldc-correct} of \cref{def:crldc}.
	Let $x \in \bin{k}$, $C = \EncIns(x)$, and let $(\pk, \sk)$ be the public and private key pair sampled by $\EncIns$ during the encoding of $x$ as $C$.
	Suppose that $\tC = C$ and let $i \in [k]$ be any index.
	We want to argue that $\Pr[\DecIns^{\tC}(i) = x_i] = 1$.
	To see this, first observe that for $(j-1)\cdot r(\lambda) < i \leq j\cdot r(\lambda)$, we have %$\Pr[\DecIns^{\tC}(i) = x_i] = \Pr[\pk^* = \pk] \cdot \Pr[ \tildex\p{j} = x\p{j} ] \cdot \Pr[ \tsigma\p{j} = \sigma\p{j} ]$, 
	\begin{align*}
		&\Pr[\DecIns^{\tC}(i) = x_i]\\ 
		&= \Pr[\pk^* = \pk] \cdot \Pr[ \tildex\p{j} = x\p{j} ] \cdot \Pr[ \tsigma\p{j} = \sigma\p{j} ],
	\end{align*}
	where $x\p{j}$ is the $j$-th block of $x$, %(\cref{line:enc-block}), 
	$\sigma\p{j}$ is the signature for $x\p{j}\circ j$, %computed as in \Cref{eq:sign} of \cref{alg:encoder}, 
	$\pk^*$ is the public key recovered via majority, % as in \cref{line:pk-majority} of \cref{alg:decoder}, 
	and $\tildex\p{j}, \tsigma\p{j}$ are parsed from the output of $\NBS(j)$. % from \cref{line:tm-parse} of \cref{alg:decoder}.
	First notice that since $\tC = C$, every block $j \in [d]$ of $\tC$ is $0$-good with respect to $C$.
	By \cref{lem:block-decode}, for every $j \in [d]$ we have that %$\Pr[ \BlockDecode^{\tC}(i) = C\p{j}] = 1$ for $i \in \phi_0(j)$.
	\begin{align*}
		\Pr_{i \in \phi_0(j)}\left[ \BlockDecode^{\tC}(i) = C\p{j} \right] = 1.
	\end{align*}
	This implies that for every $\kappa \in [\mu]$, we have 
	\begin{align*}
		\Pr[\pk\p{i_{\kappa}} = \pk]=1,
	\end{align*} 
	which implies 
	\begin{align*}
		\Pr[\pk^* = \pk]=1.
	\end{align*}
	Next, since $\tC = C$ by \cref{lem:nbs} we have that 
	\begin{align*}
		\Pr_{\tildem\p{j} \gets \NBS^{\tC}(j)}\left[\tildem\p{j} = x\p{j} \circ \sigma\p{j} \circ \pk \circ j\right] = 1.
	\end{align*} %$\tildem\p{j} = x\p{j} \circ \sigma\p{j} \circ \pk \circ j$ with probability $1$ for $\tildem\p{j} \gets \NBS^{\tC}(j)$.
	This implies that $\Verify_{\pk^*}(x\p{j} \circ j, \sigma\p{j}) = 1$ with probability $1$, yielding $\Pr[\DecIns^{\tC}(i) = x_i]=1$.
	
	We now work towards proving \cref{item:rldc-fool,item:rldc-limit}. % of \cref{def:crldc}.
	Let $\cA$ be a PPT adversary and let $\tC \gets \cA(C)$ such that $\tC \in \bin{K'}$ for some $K'$ and $\ED(\tC,C) \leq \rho$.
	We begin with \cref{item:rldc-fool}.
	Let $D_i \defeq \DecIns^{\tC}(i)$ denote the random variable of running the decoder with input $i$ and oracle $\tC$.
	Our goal is to show that %$\Pr[ \exists i \in [k] \colon \Pr[ D_i \in \{x_i,\bot\} ] < p ] \leq \eps_{\F}(\lambda)$,
	\begin{align*}
		\Pr[ \exists i \in [k] \colon \Pr[ D_i \in \{x_i,\bot\} ] < p ] \leq \eps_{\F}(\lambda),
	\end{align*}
	where $\eps_{\F}(\lambda)$ is some negligible function.
	Equivalently stated, we want to show that %$\Pr[ \forall i \in [k] \colon \Pr[ D_i \in \{x_i,\bot\} ] \geq p ] \geq 1 - \eps_{\F}(\lambda)$.
	\begin{align*}
		\Pr[ \forall i \in [k] \colon \Pr[ D_i \in \{x_i,\bot\} ] \geq p ] \geq 1 - \eps_{\F}(\lambda).
	\end{align*}
	
	We now directly analyze $ \Pr[ D_i \in \{x_i,\bot\} ]$.
	The analysis here is almost identical to the analysis of \cref{thm:ham-crldc}, except now we must take into consideration the algorithms \BlockDecode and \NBS.
	Let $\Forge_i$ denote the event that $\cA$ produces a signature forgery $(\tildex\p{j}\circ j, \tsigma\p{j})$ for $j\in [d]$ satisfying $(j-1)\cdot r(\lambda) < i \leq j\cdot r(\lambda)$ and $\tildex\p{j}$ and $\tsigma\p{j}$ are recovered from the output of $\NBS(j)$. %in \cref{line:tm-parse} of \cref{alg:decoder}.
	Then we have that %$\Pr[D_i \in \{x_i,\bot\}] = \Pr[ D_i \in \{x_i,\bot\} ~|~ \overline{\Forge_i}]$.
	\begin{align*}
		\Pr[D_i \in \{x_i,\bot\}] = \Pr[ D_i \in \{x_i,\bot\} ~|~ \overline{\Forge_i}].
	\end{align*}
	Since the decoder can never output $\bot$ and $x_i$ simultaneously, we have that %$\Pr[ D_i \in \{x_i,\bot\} | \overline{\Forge_i}] = \Pr[ D_i = x_i | \overline{\Forge_i} ] + \Pr[ D_i = \bot | \overline{\Forge_i}]$
	\begin{align*}
		&\Pr[ D_i \in \{x_i,\bot\} ~|~ \overline{\Forge_i}]\\ 
		&= \Pr[ D_i = x_i ~|~ \overline{\Forge_i} ] + \Pr[ D_i = \bot ~|~ \overline{\Forge_i}].
	\end{align*}
	We first analyze $\Pr[ D_i = x_i ~|~ \overline{\Forge_i} ]$.
	First notice that in this case, we have
	\begin{enumerate}
		\item $\pk^* \neq \bot$ and $\tildem\p{j} \neq \bot$; and
		\item $\Verify_{\pk^*}(\tildex\p{j} \circ j, \tsigma\p{j}) = 1$.
	\end{enumerate}
	Since we assume that $(\tildex\p{j} \circ j, \tsigma\p{j})$ is not a forgery, given the above it must be the case that 
	\begin{enumerate}
		\item $\tildex\p{j} = x\p{j}$ and $\tsigma\p{j} = \sigma\p{j}$; and
		\item $\pk^* = \pk$.
	\end{enumerate}
	Independent runs of $\NBS$ and $\BlockDecode$ implies %$\Pr[ D_i = x_i | \overline{\Forge_i} ] = \Pr[ \pk^* = \pk \band \tildem\p{j} \neq \bot ] =  \Pr[\pk^* = \pk] \cdot \Pr[\tildem\p{j} \neq \bot]$.
	\begin{align*}
		\Pr[ D_i = x_i ~|~ \overline{\Forge_i} ] &= \Pr[ \pk^* = \pk \band \tildem\p{j} \neq \bot ]\\ 
		&=  \Pr[\pk^* = \pk] \cdot \Pr[\tildem\p{j} \neq \bot].
	\end{align*}

	Next we analyze $\Pr[ D_i = \bot | \overline{\Forge_i}]$.
	The decoder outputs bot if $\pk^* = \bot$ or $\tildem\p{j} = \bot$ or $\Verify_{\pk^*}(\tildex\p{j}\circ j, \tsigma\p{j}) = 0$, noting that the final verification is only checked conditioned on $\pk^* \neq \bot$ and $\tildem\p{j} \neq \bot$.
	Let $\cE_1$ denote the event that $\pk^* = \bot$ and let $\cE_2$ denote the event that $\tildem\p{j} = \bot$.
	Then we have %$\Pr[ D_i = \bot | \overline{\Forge_i}] = \Pr[\pk^* = \bot] + \Pr[\tildem\p{j} = \bot]- \Pr[\pk^* = \bot \band \tildem\p{j} = \bot] +\Pr[\Verify_{\pk^*}(\tildex\p{j}\circ j, \tsigma\p{j}) = 0 | \pk^* \neq \bot \band \tildem\p{j}\neq \bot \band \overline{\Forge_i}] \cdot \Pr[\pk^* \neq \bot \band \tildem\p{j}\neq \bot]$.
	\begin{align*}
		&\Pr[ D_i = \bot ~|~ \overline{\Forge_i}]\\
		&= \Pr[\cE_1] + \Pr[\cE_2]- \Pr[\cE_1 \band \cE_2] \\
		&+\Pr\left[\Verify_{\pk^*}(\tildex\p{j}\circ j, \tsigma\p{j}) = 0 \bigg| 
		\begin{array}{l}
			\overline{\cE}_1 \band\\
			\overline{\cE}_2 \band\\
			\overline{\Forge_i}
		\end{array} \right] \cdot \Pr[\overline{\cE}_1 \band \overline{\cE}_2]. 
	\end{align*}
	Since we assume no forgery has occurred, the last summation term of the above probability can be lower bounded as %$\Pr[\Verify_{\pk^*}(\tildex\p{j}\circ j, \tsigma\p{j}) = 0 | \pk^* \neq \bot \band \tildem\p{j}\neq \bot \band \overline{\Forge_i}] \cdot \Pr[\pk^* \neq \bot \band \tildem\p{j}\neq \bot] \geq\Pr[\Verify_{\pk^*}(\tildex\p{j}\circ j, \tsigma\p{j}) = 0 | \overline{\Forge_i} \band (\tildex\p{j}\neq x\p{j}\bor\tsigma\p{j}\neq\sigma\p{j}) \band \pk^* = \pk] \cdot \Pr[\tildem\p{j}\neq \bot \band \pk^* = \pk] = \Pr[\tildem\p{j}\neq \bot \band \pk^* = \pk]$,
	\begin{align*}
		&\Pr\left[\Verify_{\pk^*}(\tildex\p{j}\circ j, \tsigma\p{j}) = 0 \bigg|
		\begin{array}{l}
			\overline{\cE}_1 \band\\
			\overline{\cE}_2 \band\\
			\overline{\Forge_i}
		\end{array} \right] \cdot \Pr[\overline{\cE}_1 \band \overline{\cE}_2]\\
		&\geq \Pr\left[ \Verify_{\pk^*}(\tildex\p{j}\circ j, \tsigma\p{j}) = 0 \bigg| \begin{array}{l}
			\overline{\Forge_i} \band\\
			(\tildex\p{j}\neq x\p{j}\bor\\
			\tsigma\p{j}\neq\sigma\p{j}) \band\\
			\pk^* = \pk
		\end{array} \right]\\
		&\hspace*{2em}\cdot \Pr[\overline{\cE}_2 \band \pk^* = \pk]\\
		&=\Pr[\overline{\cE}_2 \band \pk^* = \pk],
	\end{align*}
	where for the lower bound we ignore the event that $\pk^* \neq \pk \band \pk^* \neq \bot$.
	Thus we have that $\Pr[ D_i \in \{x_i,\bot\} | \overline{\Forge_i}] \geq \Pr[\pk^* = \bot] + \Pr[\tildem\p{j} = \bot] - \Pr[\pk^* = \bot]\cdot \Pr[\tildem\p{j} = \bot] + \Pr[\tildem\p{j} \neq \bot]\cdot \Pr[\pk^* = \pk]$, 
	\begin{align*}
		\Pr[ D_i \in \{x_i,\bot\} ~|~ \overline{\Forge_i}] &\geq \Pr[\cE_1] + \Pr[\cE_2] - \Pr[\cE_1]\cdot \Pr[\cE_2] \\ 
		&+ \Pr[\overline{\cE}_2]\cdot \Pr[\pk^* = \pk]
	\end{align*}
	where the inequality holds since \NBS and \BlockDecode are run independently of each other.
	
	With the above lower bound established, we turn to analyzing $\Pr[\pk^* = \pk]$.
	Let $\pk$ be the public key sampled by $\EncIns(x)$ to generate codeword $C$.
	Recovery of $\pk$ is performed by uniformly sampling $\mu$ indices of $\tC$, running $\BlockDecode$ on these indices, and taking the majority of the public keys we recover. 
	Intuitively, we want to recover $\pk^*$ with high probability via a Chernoff bound.
	Define random variable $X_\kappa$ for $\kappa \in [\mu]$ as %$X_\kappa = 1$ if $x\p{j}\circ\sigma\p{j}\circ\pk\circ j = \BlockDecode(i_\kappa)$ for $j \in [d]$ and $X_\kappa = 0$ otherwise.
	\begin{align*}
		X_\kappa \defeq \begin{cases}
			1 & x\p{j}\circ\sigma\p{j}\circ\pk\circ j = \BlockDecode(i_\kappa) \text{ for $j \in [d]$}  \\
			0 & \text{otherwise}
		\end{cases}.
	\end{align*}
	Thus we need to ensure that $\Pr[X_\kappa = 1] > 1/2$.
	By \cref{lem:block-decode}, we know that as long as the index $i_\kappa$ lies within the bounds of some $\gamma$-good block $j$, then $X_\kappa = 1$ with probability at least $1-\gamma$.
	Let $\cJ \subset [d]$ be the indices of the $\gamma$-good blocks in $\tC$.
	Then we have %$\Pr[X_\kappa = 1] \geq \Pr[i_\kappa \in \phi_0^{-1}(j) | j \in \cJ] \cdot (1-\gamma)$.
	\begin{align*}
		\Pr[X_\kappa = 1] \geq \Pr[i_\kappa \in \phi_0^{-1}(j) ~|~ j \in \cJ_{\Good}] \cdot (1-\gamma)
	\end{align*}
	By \cref{lem:gamma-block-bounds}, we know that there are at least $(1-2 \beta \rho/(\gamma \alpha))$-fraction of blocks which are $\gamma$-good, which implies $|\cJ| \geq d \cdot(1- 2 \beta \rho/(\gamma \alpha))$.
%	Note also that the notion of $\gamma$-good is only useful whenever $\gamma \leq \rhosz$, else we would not be able to decode correctly.
	Since we want $\Pr[X_\kappa=1] > 1/2$, by \cref{lem:block-decomp-bound} we have 
	\begin{align*}
		(d/n)\cdot(1-(2\beta\rho)/(\gamma\alpha))\cdot(\beta-\alpha\gamma)\cdot\tau\cdot(1-\gamma) >1/2.
	\end{align*}
	Solving for $\rho$ above yields %$\rho < ((\gamma\alpha)/(2\beta))\cdot(1-\beta/(2(1-\gamma)(\beta-\alpha\gamma)))$.
	\begin{align*}
%		\frac{1}{n} \cdot d \cdot \left(1-\frac{2 \cdot \beta \cdot \rho}{\gamma \cdot \alpha}\right) \cdot (\beta-\alpha\cdot \gamma)\cdot \tau \cdot (1-\gamma) > \frac{1}{2}\\
%		\frac{(\beta-\alpha\cdot \gamma) \cdot \tau}{\beta\cdot \tau} \cdot \left(1-\frac{2 \cdot \beta \cdot \rho}{\gamma \cdot \alpha}\right) > \frac{1}{2\cdot(1-\gamma)}\\
%		\left(1-\frac{2 \cdot \beta \cdot \rho}{\gamma \cdot \alpha}\right) > \frac{\beta}{2 \cdot (1-\gamma) \cdot (\beta - \alpha \cdot \gamma)}\\
%		\frac{2\cdot \beta \cdot \rho}{\gamma \cdot \alpha} < 1 - \frac{\beta}{2 \cdot (1-\gamma) \cdot (\beta - \alpha \cdot \gamma)}\\
		\rho < \frac{\gamma\cdot \alpha}{2\cdot \beta}\cdot \left(1 - \frac{\beta}{2 \cdot (1-\gamma) \cdot (\beta - \alpha \cdot \gamma)}\right).
	\end{align*}
	Borrowing the parameters of Block \etal \cite[Proposition 22]{BBGKZ20}, we set $\gamma = 1/12$ and $\alpha = 2\gamma \rhoin/(\gamma+6)$.
	Recall also that $\beta = 2\alpha + (1/\betasz)$.
	Then we have %$\alpha = (2/73) \cdot \rhoin$, $\beta = (4/73) \cdot \rhoin + (1/\betasz)$, $(\gamma\alpha)/(2\beta) = (2\rhoin/876) / (2/\betasz + 8\rhoin/73) < 1$, and $\beta/(2(1-\gamma)(\beta-\alpha\gamma)) = (1/\betasz + 4\rhoin/73) / ( (11/6)\cdot ( 1/\betasz +  (23\rhoin/438) ) ) < 1$.
	\begin{gather*}
		\alpha = (2/73) \cdot \rhoin\\
		\beta = (4/73) \cdot \rhoin + (1/\betasz)\\
		\frac{\gamma\cdot \alpha}{2\cdot \beta} = \frac{(1/12)\cdot(2/73)\cdot \rhoin}{2\cdot(4/73)\cdot \rhoin + 2/\betasz} < 1\\
		\frac{\beta}{2(1-\gamma)(\beta-\alpha\cdot\gamma)} = \frac{6}{11} \cdot \frac{(4/73) \rhoin + (1/\betasz)}{(23/438) \rhoin + (1/\betasz)} < 1.
	\end{gather*}
	Given these above parameters, since $\rhoin < 1$ we also have %$(\gamma\alpha)/(2\beta) \geq (2\rhoin/876) / (2/\betasz + 8/73)$.
	\begin{align*}
		\frac{\gamma \cdot \alpha}{2 \cdot \beta} \geq \frac{(2/876)\cdot \rhoin}{(8/73) + 2/\betasz}.
	\end{align*}
	Thus setting %$\rho \defeq [(2\rhoin/876) / (2/\betasz + 8/73)]\cdot[1-(1/\betasz + 4\rhoin/73) / ( (11/6)\cdot ( 1/\betasz +  (23\rhoin/438) ) )] = \Theta(1)$
	\begin{align*}
		\rho &\defeq \\
		&\left(\frac{(2/876) \rhoin}{(8/73) + 2/\betasz}\right) \left(1- \frac{6}{11}\cdot\frac{(4/73) \rhoin + (1/\betasz)}{(23/438)\cdot \rhoin + (1/\betasz)} \right)\\
		&= \Theta(1)
	\end{align*}
	ensures that $\Pr[X_\kappa = 1] > 1/2$.
	Now let $q \defeq \Pr[X_\kappa = 1] > 1/2$.
	By a Chernoff bound we have that %$\Pr[\sum_\kappa X_\kappa > \mu/2] \geq 1 - \exp(-\mu (q-1/2)^2/(2q))$.
	\begin{align*}
		\Pr\left[ \sum_{\kappa \in [\mu]} X_\kappa > \frac{\mu}{2} \right] \geq 1 - \exp(-\mu \cdot (q-1/2)^2/(2q)).
	\end{align*}
	This implies $\pk^* = \pk$ with probability at least $1 - \exp(-\mu (q-1/2)^2/(2q))$,
%	\begin{align*}
%		1 - \exp(-\mu (q-1/2)^2/(2q)),
%	\end{align*} 
%	we have that $\pk^* = \pk$, 
	which implies 
	\begin{align*}
		\Pr[\pk^* = \bot] \leq \exp(-\mu(q-1/2)^2/(2q)).
	\end{align*}
	
	Given that 
	\begin{align*}
		\Pr[\pk^* = \pk] \geq 1 - \exp(-\mu (q-1/2)^2/(2q)),
	\end{align*} 
	we now analyze $\Pr[\tildem\p{j} \neq \bot]$.
	Note that by our noisy binary search algorithm in \cref{lem:nbs}, if block $j$ is $\gamma$-good, then we recover the correct block $x\p{j}\circ \sigma\p{j}\circ \pk \circ j$ with probability at least $(1-\rho^*)$, except with probability $\vartheta(K')$, where $\rho^*$ is some constant and $\vartheta$ is a negligible function.
	Conditioning on the case where we recover any $\gamma$-good block with probability at least $(1-\rho^*)$ we have %$\Pr[ \tildem\p{j} \neq \bot | j \text{ is $\gamma$-good} ] \geq (1-\rho^*)$; equivalently
	\begin{align*}
		\Pr[ \tildem\p{j} \neq \bot ~|~ j \text{ is $\gamma$-good} ] \geq (1-\rho^*).
	\end{align*}
	Equivalently, we have
%	$\Pr[ \tildem\p{j} = \bot | j \text{ is $\gamma$-good} ] \leq \rho^*$.
	\begin{align*}
		\Pr[ \tildem\p{j} = \bot ~|~ j \text{ is $\gamma$-good} ] \leq \rho^*.
	\end{align*}
	Putting it all together we have that %$\Pr[D_i \in \{x_i,\bot\} | \overline{\Forge_i}] \geq (1-\rho^*)\cdot(1-\exp(-\mu \cdot (q-1/2)^2/(2q))) - \rho^* \cdot \exp(-\mu \cdot (q-1/2)^2/(2q)) = 1-\rho^*-\exp(-\mu \cdot (q-1/2)^2/(2q))$,
	\begin{align*}
		&\Pr[D_i \in \{x_i,\bot\} ~|~ \overline{\Forge_i}]\\
		&\geq (1-\rho^*)\cdot\left(1-\exp\left(-\mu \cdot \frac{(q-1/2)^2}{2q}\right)\right)\\ 
		&- \rho^* \cdot \exp\left(-\mu \cdot \frac{(q-1/2)^2}{2q}\right)\\
		&= 1-\rho^*-\exp\left(-\mu \cdot \frac{(q-1/2)^2}{2q}\right),
	\end{align*}
	where $\rho^* \in (0,1/2)$ is a fixed constant we are free to choose.
	In particular, we choose $\rho^* < 1/3$ to help ensure $p > 2/3$ for appropriate choice of $\mu$.
	
	In the above analysis, we conditioned on $\overline{\Forge_i}$ and on \NBS succeeding.
	By the security of the digital signature scheme, $\overline{\Forge_i}$ occurs with probability at least $1-\eps_{\Pi}(\lambda)$, where $\eps_{\Pi}(\lambda)$ is a negligible function for the security of the digital signature scheme.
	By \cref{lem:nbs}, with probability at least $1-\vartheta(n')$, \NBS outputs correctly.
	Finally, by union bound over $i \in [k]$, we set $\eps_{\F}(\lambda, n) \defeq k \cdot \eps_{\Pi}(\lambda) \cdot \vartheta(n)$, which is negligible in $\lambda$ since $k(\lambda) = \poly(\lambda)$.

	Next we work towards proving \cref{item:rldc-limit}.
	Again let $\tC \gets \cA(C)$ for a PPT adversary $\cA$ and $C = \EncIns(x)$ for $x \in \bin{k}$.
	Our goal is to show that there exists a negligible function $\eps_{\Lim}(\cdot)$ such that for all $\lambda \in \bbN$ and all $x \in \bin{k}$, we have %$\Pr[\Limit(\cA(C), \rho, \delta, x, y) = 1] \leq \eps_{\Lim}(\lambda)$.
	\begin{align*}
		\Pr[\Limit(\cA(C), \rho, \delta, x, y) = 1] \leq \eps_{\Lim}(\lambda).
	\end{align*}
	We directly analyze the size of the set $\Good(\tC)$.
	As before, let $D_i$ be the random variable denoting the output of $\DecIns^{\tC}(i)$.
	Then we are interested in lower bounding $\Pr[D_i = x_i]$ for a fixed $i \in [k]$.
	By definition of $\DecIns$, the decoder only outputs a bit if
	\begin{enumerate}
		\item $\pk^* \neq \bot$;
		\item $\tildem\p{j} \neq \bot$; and
		\item $\Verify_{\pk^*}(\tildex\p{j}\circ j, \tsigma\p{j}) = 1$.
	\end{enumerate} 
	As before, let $\cE_1$ denote the event that $\pk^* = \bot$ and $\cE_2$ denote the event that $\tildem\p{j} = \bot$.
	Then by definition we have
	\begin{align*}
		\Pr[D_i = x_i] = &\Pr[\overline{\cE}_1] \cdot \Pr[\overline{\cE}_2 | \overline{\cE}_1]\\ 
		&\cdot \Pr[\Verify_{\pk^*}(\tildex\p{j} \circ j, \tsigma\p{j}) = 1 ~|~ \overline{\cE}_1\band \overline{\cE}_2].
	\end{align*}
%	$\Pr[D_i = x_i] = \Pr[\pk^* \neq \bot] \cdot \Pr[\tildem\p{j}\neq \bot | \pk^* \neq \bot] \cdot \Pr[\Verify_{\pk^*}(\tildex\p{j} \circ j, \tsigma\p{j}) = 1 | \pk^* \neq \bot \band \tildem\p{j}\neq \bot]$ by definition.
	Since $\pk^*$ and $\tildem\p{j}$ are computed independently, we have the above equals %$\Pr[\pk^* \neq \bot] \cdot \Pr[\tildem\p{j}\neq \bot] \cdot \Pr[\Verify_{\pk^*}(\tildex\p{j} \circ j, \tsigma\p{j}) = 1 | \pk^* \neq \bot \band \tildem\p{j}\neq \bot]$, which is at least $\Pr[\pk^* = \pk] \cdot \Pr[ \tildem\p{j} = m\p{j} ] \cdot \Pr[\Verify_{\pk^*}(\tildex\p{j} \circ j, \tsigma\p{j}) = 1 | \pk^* = \pk \band \tildem\p{j} = m\p{j}]$ (the lower bound is a subset of events in consideration).
	\begin{align*}
		&\Pr[D_i = x_i]\\
		&= \Pr[\overline{\cE}_1] \cdot \Pr[\overline{\cE}_2] \cdot \Pr[\Verify_{\pk^*}(\tildex\p{j} \circ j, \tsigma\p{j}) = 1 ~|~ \overline{\cE}_1 \band \overline{\cE}_2]\\
		&\geq \Pr[\pk^* = \pk] \cdot \Pr[ \tildem\p{j} = m\p{j} ] \\
		&\cdot \Pr[\Verify_{\pk^*}(\tildex\p{j} \circ j, \tsigma\p{j}) = 1 ~|~ \pk^* = \pk \band \tildem\p{j} = m\p{j}], %\label{eq:dec-correct-bit-lb}
	\end{align*}
	where the inequality follows since it is a subset of all possible events in consideration. and $m\p{j} \defeq x\p{j} \circ \sigma\p{j} \circ \pk \circ j$. %, \cref{eq:pk-mj-indep} follows since $\pk^*$ and $\tildem\p{j}$ are generated independently, and \cref{eq:dec-correct-bit-lb} follows since these events are a subset of the events in \cref{eq:pk-mj-indep}.
	Now we analyze each of the terms in the above lower bound.
	First note that the final term 
	\begin{align*}
		\Pr[\Verify_{\pk^*}(\tildex\p{j} \circ j, \tsigma\p{j}) = 1 ~|~ \pk^* = \pk \band \tildem\p{j} = m\p{j}] = 1
	\end{align*} 
	by definition.
	Thus we analyze the other two probability terms.
	By our prior work, we know that 
	\begin{align*}
		\Pr[\pk^* = \pk] \geq 1 - \exp(-\mu(q-1/2)^2/(2q)).
	\end{align*}

%	We begin with $\Pr[\pk^* = \pk]$.
%	We can lower bound this probability via a Chernoff bound.
	%To do so, first we do some preliminary analysis on the behavior of the decoder.
	%Let $D_i$ be the random variable representing the output of $\Dec^{\tC}(i)$.
	%We are interested in the following probability:
	%\begin{align*}
	%	\Pr[ D_i \in \{x_i, \bot\} ] = \Pr[D_i = x_i] + \Pr[D_i = \bot].
	%\end{align*}
	%We first analyze $\Pr[D_i = x_i]$, which is equal to the following quantity.
	%\begin{align*}
	%	\Pr[D_i = x_i] &\geq \Pr[\pk^* = \pk \band \tildem\p{j} = C\p{j} \band \Verify_{\pk^*}(\tildex\p{j} \circ j, \tsigma\p{j}) = 1]\\
	%	&= \Pr[\pk^* = \pk] \cdot \Pr[\tildem\p{j} = C\p{j}~|~\pk^* = \pk]\\ 
	%	&\cdot \Pr[ \Verify_{\pk^*}(\tildex\p{j}\circ j, \tsigma\p{j}) = 1 ~|~  \pk^* = \pk \band \tildem\p{j} = C\p{j}]\\
	%	&= \Pr[\pk^* = \pk] \cdot \Pr[ \tildem\p{j} = C\p{j} ] \cdot \Pr[ \Verify_{\pk^*}(\tildex\p{j}\circ j, \tsigma\p{j}) = 1 ~|~  \pk^* = \pk \band \tildem\p{j} = C\p{j}].
	%\end{align*}
	%Above, the second equality follows from the observation that $\tildem\p{j}$ and $\pk^*$ are independent.
	%We can lower bound $\Pr[\pk^* = \pk]$ via a Chernoff bound.

	We now lower bound $\Pr[\tildem\p{j} = m\p{j}]$.
	Recall that $\tildem\p{j} = \NBS^{\tC}(j)$. 
	By \cref{lem:nbs} we have %$\Pr[\Pr[\tildem\p{j} \neq C\p{j} | j \text{ is } \gamma\text{-good} ]\geq \rho^*]\leq \vartheta(K')$.
	\begin{align*}
		\Pr\left[ \Pr_{j \getsr [d]}\left[\tildem\p{j} \neq C\p{j}~|~ j \text{ is $\gamma$-good}\right] \geq \rho^* \right] \leq \vartheta(K').
	\end{align*}
	Fix $j$ to be a $\gamma$-good block. Then we have %$\Pr[ \Pr[\tildem\p{j} \neq C\p{j}] \geq \rho^* ] \leq \vartheta(K')$, 
	\begin{align*}
		\Pr\left[ \Pr[\tildem\p{j} \neq C\p{j}] \geq \rho^* \right] \leq \vartheta(K');
	\end{align*}
	or equivalently %$\Pr\left[ \Pr[\tildem\p{j} = C\p{j}] \leq 1-\rho^* \right] \leq \vartheta(K')$.
	\begin{align*}
		\Pr\left[ \Pr[\tildem\p{j} = C\p{j}] \leq 1-\rho^* \right] \leq \vartheta(K').
	\end{align*}
	This implies %$\Pr[ \Pr[ \tildem\p{j} = C\p{j} ] \geq 1-\rho^* ] \geq 1 - \vartheta(K')$.
	\begin{align*}
		\Pr\left[ \Pr[ \tildem\p{j} = C\p{j} ] \geq 1-\rho^* \right] \geq 1 - \vartheta(K').
	\end{align*}
	Thus with probability at least $1-\vartheta(K')$, we have 
	\begin{align*}
		\Pr[\tildem\p{j} = C\p{j}] \geq 1-\rho^*.
	\end{align*}
	By \cref{lem:gamma-block-bounds}, we know at least $1-(2 \beta \rho)/(\gamma\alpha)$-fraction of blocks in $\tC$ are $\gamma$-good.
	Since for every $i \in [k]$ there is a unique block $j\in[d]$ such that $(j-1) r(\lambda) < i \leq j  r(\lambda)$, for the set 
	\begin{align*}
		\cG \defeq \{i \in [k] \colon i \text{ is in a } \gamma\text{-good block}\},
	\end{align*} 
	we have %$|\cG| \geq r(\lambda) \cdot ( 1-(2\beta\rho)/(\gamma\alpha)) \cdot d = (1-(2\beta\rho)(\gamma\alpha)) \cdot k$; so
	\begin{align*}
		|\cG| &\geq r(\lambda) \cdot \left( 1-\frac{2\beta\rho}{\gamma\alpha} \right) \cdot d= \left( 1-\frac{2\beta\rho}{\gamma\alpha} \right) \cdot k.
	\end{align*}
	Therefore at least $1-(2 \beta \rho)/(\gamma\alpha)$-fraction of indices $i \in [k]$ lie within a $\gamma$-good block.
	
	Putting it all together, for any $\gamma$-good block we have that with probability at least $1-\vartheta(K')$: %$\Pr[D_i = x_i] \geq (1-\exp(-\mu(q-1/2)^2/(2q))) \cdot (1-\rho^*)$.
	\begin{align*}
		\Pr[D_i = x_i] \geq (1-\exp(-\mu\cdot(q-1/2)^2/(2q))) \cdot (1-\rho^*).
	\end{align*}
	Choosing $\mu$ and $\rho^* \in (0,1/3)$ appropriately such that 
	\begin{align*}
		(1-\exp(-\mu(q-1/2)^2/(2q))) \cdot (1-\rho^*) > 2/3,
	\end{align*} 
	along with the fact that at least $1-(2 \beta \rho)/(\gamma\alpha)$-fraction of indices $i \in [k]$ lie within a $\gamma$-good block, we have 
	\begin{align*}
		\Pr[|\Good(\tC)| \geq (1-(2 \beta \rho)/(\gamma \alpha)) \cdot k] \geq 1-\vartheta(K'),
	\end{align*} 
	which implies 
	\begin{align*}
		\Pr[\Limit(\tC, \rho, \delta, x, C) = 1] &\leq \vartheta(K')\\
		&= \vartheta(\Theta(k)) \eqdef \eps_{\Lim}(\lambda)
	\end{align*}   
%	\begin{align*}
%		|\Good(\tC)| \geq \left(1-\frac{2 \cdot \beta \cdot \rho}{\gamma \cdot \alpha}\right) \cdot k.
%	\end{align*}
	for $\delta = 1-(2 \beta \rho)/(\gamma \alpha) = 1-\Theta(\rho)$. 
%	\begin{align*}
%		\Pr[\Limit(\tC, \rho, \delta, x, C) = 1] \leq \vartheta(n').
%	\end{align*}
%	As $\rho = \Omega(\rho^* \cdot \rhosz)$ is a constant and $n = \Theta(k \cdot [3+\log(k)/r(\lambda)])$, we have that $\vartheta(n') = \negl(k)$, which is $\negl(\lambda)$ whenever $k = \poly(\lambda)$.
\end{proof}

%\input{construction}

% !TeX root = ../full-main.tex
\section{Acknowledgments}\label{sec:ack}
Alexander R. Block was supported in part by the National Science Foundation under NSF CCF \#1910659 and in part by DARPA under agreement No. HR00112020022 and No. HR00112020025. Jeremiah Blocki was supported in part by the National Science Foundation under awards CNS \#2047272, CNS \#1931443 and CCF \#1910659. The views, opinions, findings, conclusions and/or recommendations expressed in this material are those of the author and should not be interpreted as reflecting the position or policy of the Department of Defense or the U.S. Government, and no official endorsement should be inferred.

%%% bibliography %%%

%% for normal articles %%
%\bibliographystyle{alphaurl} % for normal articles
%\bibliography{local}

%\newpage
\bibliographystyle{alphaurl}
\bibliography{full-version/abbrev3,full-version/crypto,full-version/local}

%\printbibliography

%% for llncs %%
%\renewcommand{\doi}[1]{\url{https://doi.org/#1}} % fix LNCS' issue where DOI with underscores break bib compilation; aka, LNCS is broken and dumb
%\bibliographystyle{splncs04}

%% for IEEE articles %%
%\bibliographystyle{IEEEtran} 

%% lipics has nothing as it's included in the class %%

%\bibliography{} % bib files

%\newpage
%\appendix
%\section{Scratch Work}
%\input{alg-ham-encoder}

%\input{alg-ham-decoder}

\end{document}